\documentclass[12pt]{article}

\RequirePackage[OT1]{fontenc}
\RequirePackage{amsthm,amsmath,amssymb}
\usepackage{libertine}
\usepackage[libertine]{newtxmath}
\RequirePackage{natbib}
\RequirePackage[colorlinks,citecolor=blue,urlcolor=blue]{hyperref}
\usepackage{authblk}
\usepackage{bbm}
\usepackage{multirow}
\usepackage{paralist}
\usepackage[left=1.375in, right=1.375in]{geometry}
\usepackage{graphicx}
\usepackage{mathtools}
\usepackage[bf]{caption}
\usepackage{rotating}
\usepackage{pgfplots}
\usepackage{makecell}
\usepackage{comment}
\mathtoolsset{mathic=true}
\usepackage{filecontents}

\theoremstyle{plain}
\newtheorem{thm}{Theorem}[section]
\newtheorem{lemma}{Lemma}[section]

\theoremstyle{definition}
\newtheorem{definition}{Definition}[section]
\newtheorem{procedure}{Procedure}[section]

\numberwithin{equation}{section}
\numberwithin{table}{section}

\DeclareMathOperator{\cov}{Cov}

\newcommand{\ex}{\mathrm{E}}
\newcommand{\pr}{\mathrm{P}}

\newcommand{\rot}{{\scriptscriptstyle\circlearrowleft}}
\newcommand{\srv}{\mathrm{s}}
\newcommand{\convl}{\rightsquigarrow}

\makeatletter
\newcommand{\oset}[3][0ex]{
	\mathrel{\mathop{#3}\limits^{
			\vbox to#1{\kern0.5\ex@
				\hbox{$\scriptstyle#2$}\vss}}}}

\newcommand{\uset}[3][0ex]{%
	\mathrel{\mathop{#3}\limits^{
			\vbox to#1{\kern7\ex@
				\hbox{$\scriptstyle#2$}\vss}}}}

\newcommand{\cwconv}{\oset{\pr}{\uset{\tau}{\rightsquigarrow}}}
\makeatother

\makeatletter
\newcommand{\osetd}[3][0ex]{%
	\mathrel{\mathop{#3}\limits^{
			\vbox to#1{\kern-5\ex@
				\hbox{$\scriptstyle#2$}\vss}}}}
\newcommand{\eql}{\osetd{\mathcal{D}}{=}}
\makeatother

\begin{document}
	
	\title{Randomization tests of copula symmetry}
	
	\author[1]{Brendan K.\ Beare}
	\author[2]{Juwon Seo}
	\affil[1]{School of Economics, University of Sydney}
	\affil[2]{Department of Economics, National University of Singapore}
	
	\maketitle
	\begin{center}
		Accepted for publication in \textit{Econometric Theory}.
	\end{center}	
	\bigskip

\begin{abstract}
	New nonparametric tests of copula exchangeability and radial symmetry are proposed. The novel aspect of the tests is a resampling procedure that exploits group invariance conditions associated with the relevant symmetry hypothesis. They may be viewed as feasible versions of randomization tests of symmetry, the latter being inapplicable due to the unobservability of margins. Our tests are simple to compute, control size asymptotically, consistently detect arbitrary forms of asymmetry, and do not require the specification of a tuning parameter. Simulations indicate excellent small sample properties compared to existing procedures involving the multiplier bootstrap.
\end{abstract}

\newpage

\section{Introduction}\label{secintro}

In this paper we propose statistical tests of the null hypothesis that a copula \(C\) is symmetric, based on a sample of independent and identically distributed (iid) pairs of random variables with common copula $C$. We focus on two notions of symmetry that have received particular attention in the literature: exchangeability and radial symmetry. Let \((U,V)\) be a pair of random variables whose joint distribution is given by the copula \(C\).  We say that \(C\) is exchangeable when
\begin{equation}
C(u,v)=C(v,u)\text{ for all }(u,v)\in[0,1]^2.
\end{equation}
Exchangeability of \(C\) is satisfied if and only if \((U,V)\eql(V,U)\), where \(\eql\) signifies equality in law. We say that \(C\) is radially symmetric when
\begin{equation}
C(u,v)=C^\srv(u,v)\text{ for all }(u,v)\in[0,1]^2,
\end{equation}
where \(C^\srv(u,v)=u+v-1+C(1-u,1-v)\), the survival copula for \(C\). Radial symmetry of \(C\) is satisfied if and only if \((U,V)\eql(1-U,1-V)\). See \citet{N93,N06,N07} for further discussion of the exchangeability and radial symmetry properties.

The property of exchangeability plays an important role in various models of economic interaction. \citet{M16} writes that ``exchangeability of a certain form is a feature of almost any commonly used empirical specification for game-theoretic models with more than two players''. A prominent example is the symmetric common value auction model, which was developed by \citet{MW82} under the assumption that the distribution of signals across bidders is exchangeable. Such exchangeability has powerful implications for the identification of structural econometric models of auctions \citep{AH02} and is frequently assumed when they are estimated \citep{LPV00,HPP03,T11}. Another example is the model of product bundling developed by \citet{CR13b}, in which the exchangeability of the copula describing the dependence between consumer valuations of different products is a central assumption when a multi-product firm competes with a single-product firm. Radial symmetry, or rather the lack thereof, has been a subject of interest in empirical finance: researchers have found that the dependence between various asset returns, particularly equity portfolios, is markedly stronger in downturns than in upturns \citep{AC02,HTZ02}. Radially asymmetric copula functions have proved to be useful for modeling this feature of return dependence \citep{P04,P06,O08,GT11}.

Several statistical tests of exchangeability and radial symmetry for bivariate copulas have been proposed in recent literature. \citet{GNQ12} and \citet{GN14} proposed tests of copula exchangeability and radial symmetry respectively. Extensions of Genest et al.'s exchangeability tests to higher dimensional copulas have been provided by \citet{HS17}. Genest and Ne\v{s}lehov\'{a}'s radial symmetry tests extend earlier contributions of \citet{BC12} and \citet{DDU13}. Other tests of copula exchangeability and radial symmetry were proposed by \citet{LG13}, \citet{QB13}, \citet{BQ17}, \citet{BIW17} and \citet{K17}. \citet{BS14} also proposed a test of copula exchangeability, but for the somewhat different case where the copula in question characterizes the serial dependence in a univariate time series. Many other authors have considered tests of exchangeability or radial symmetry for multivariate cdfs---see, for instance, \citet{Q16} and references therein---but such tests are typically inapplicable to hypotheses of copula symmetry due to the unobservability of margins.

The new tests of copula symmetry proposed in this paper combine the test statistics of \cite{GNQ12} and \citet{GN14} with a new method of constructing critical values. Whereas those authors obtain critical values using the multiplier bootstrap of \citet{RS09} and \citet{BD10}, we instead use a novel resampling procedure motivated by randomization tests of symmetry hypotheses. \citet{R89,R90} observed that exact tests of symmetry hypotheses on multivariate cdfs could be obtained by applying randomization procedures that exploit group invariance conditions implied by symmetry. While these tests are not directly applicable to hypotheses of copula symmetry, we show how a feasible randomization procedure may be used to obtain critical values that properly account for uncertainty about margins. The justification for our procedure is asymptotic rather than exact, but numerical simulations indicate excellent size control with sample sizes as small as \(n=30\). Simulations also indicate substantially improved power compared to the tests based on the multiplier bootstrap at smaller sample sizes.

A recent paper by \citet{CRS17} is related to ours in that it studies the behavior of randomization tests when symmetry is only approximately satisfied. Suppose we have a sample \(Z^{(n)}\) of size \(n\) taking values in a sample space \(\mathcal Z_n\). Approximate symmetry in the sense of \citet{CRS17} means that for each \(n\) there exists a map \(S_n\) from \(\mathcal Z_n\) to a metric space \(\mathcal S\) such that (i) \(S_n(Z^{(n)})\) converges in law to a random element \(S\) of \(\mathcal S\) as \(n\to\infty\), and (ii) \(g(S)\) is equal in law to \(S\) for all \(g\) in some finite group of transformations \(\mathbf G\). Crucially, \(\mathcal S\) and \(\mathbf G\) cannot depend on \(n\). In our paper, the sample is a collection of iid pairs \(Z^{(n)}=((X_1,Y_1),\ldots,(X_n,Y_n))\) taking values in \(\mathcal Z_n=(\mathbf R^2)^n\). Approximate symmetry holds in the following sense: if \(S_n:\mathcal Z_n\to([0,1]^2)^n\) is the map that transforms our sample to the normalized rank pairs \(((U_{n1},V_{n1}),\ldots,(U_{nn},V_{nn}))\) defined in equation \eqref{rankpairs} below, then the law of \(S_n(Z^{(n)})\) is (loosely speaking) approximately that of \(n\) iid draws from the copula \(C\). When \(C\) is symmetric, such \(n\)-tuples of iid draws are distributionally invariant under a group \(\mathbf G_n\) consisting of \(2^n\) distinct transformations from \(([0,1]^2)^n\) to itself; we postpone details of \(\mathbf G_n\) to Section \ref{secresample}. Since the dimension of \(([0,1]^2)^n\) and the number of transformations in \(\mathbf G_n\) grow with \(n\), our problem falls outside the scope of the results of \citet{CRS17}, in which \(\mathcal S\) and \(\mathbf G\) are assumed fixed.

Recent results of \citet{CR13,CR16a,CR16b} are also somewhat related to the problem studied here. Like us, and unlike \citet{CRS17}, Chung and Romano allow the number of transforms in the group \(\mathbf G_n\) to increase with \(n\). However, whereas in our setting the normalized rank pairs \(((U_{n1},V_{n1}),\ldots,(U_{nn},V_{nn}))\) are approximately distributionally invariant under \(\mathbf G_n\) whenever the null is satisfied, in the setting considered by Chung and Romano the data are exactly distributionally invariant under \(\mathbf G_n\) on a subset of the null, and not even approximately invariant elsewhere in the null. The problems we study are therefore fundamentally different. Chung and Romano establish their results by verifying a condition of \citet{H52} necessary and sufficient for suitable convergence of the randomization distribution. We instead take the conditional approach to which \citet[p.\ 497]{CR13} refer following their discussion of Hoeffding's condition. Specifically, in place of Hoeffding's condition we verify that a statistic computed from a random transformation of the normalized rank pairs converges weakly to a suitable limit conditional on the data in probability.

Our paper is structured as follows. In Section \ref{secteststats} we define our test statistics and characterize their limit distributions under the null hypothesis of symmetry. In Section \ref{secresample} we describe our feasible randomization procedure for obtaining critical values. In Section \ref{secasym} we provide results on the asymptotic properties of tests based on our feasible randomization procedure. The results of our numerical simulations are presented in Section \ref{secsim}. Some closing remarks are given in Section \ref{conclusion}. Proofs and supplementary lemmas are collected in Appendix \ref{appx}.

\section{Test statistics}\label{secteststats}

\subsection{Basic setup}

Let \(X\) and \(Y\) be random variables with bivariate cumulative distribution function (cdf) \(H(x,y)=\pr(X\leq x,Y\leq y)\) and margins \(F(x)=\pr(X\leq x)\) and \(G(y)=\pr(Y\leq y)\). We assume that \(F\) and \(G\) are continuous. Sklar's theorem \citep{S59} then ensures the existence of a unique copula \(C:[0,1]^2\to[0,1]\) satisfying \(C(F(x),G(y))=H(x,y)\) for all \(x,y\in\mathbf R\). The copula \(C\) is the bivariate cdf of the probability integral transforms \(U=F(X)\) and \(V=G(Y)\).

Our data consist of \(n\) iid draws \((X_1,Y_1),\ldots,(X_n,Y_n)\) from \(H\). Let \(F_n\), \(G_n\) and \(H_n\) be the empirical cdfs corresponding to \(F\), \(G\) and \(H\) respectively. We use \(F_n\) and \(G_n\) to construct (normalized) ranks
\begin{equation}\label{rankpairs}
U_{ni}=F_n(X_i),\quad V_{ni}=G_n(Y_i),\quad i=1,\ldots,n.
\end{equation}
From the rank pairs \((U_{ni},V_{ni})\) we construct the empirical copula
\begin{equation}\label{empcop}
C_n(u,v)=\frac{1}{n}\sum_{i=1}^n\mathbbm{1}\left(U_{ni}\leq u,V_{ni}\leq v\right),\quad(u,v)\in[0,1]^2.
\end{equation}
An alternative definition of the empirical copula in common use is
\begin{equation}\label{empcopD}
C_n^\mathrm{D}(u,v)=H_n\left(F_n^\leftarrow(u),G_n^\leftarrow(v)\right),\quad(u,v)\in[0,1]^2,
\end{equation}
where \(F_n^\leftarrow\) is the generalized inverse of \(F_n\),
\begin{eqnarray}
F_n^\leftarrow(u)&=&\inf\{x\in\mathbf R:F_n(x)\geq u\}, \quad u\in(0,1],\\
F_n^\leftarrow(0)&=&F_n^\leftarrow(0+),
\end{eqnarray}
and \(G_n^\leftarrow\) is defined similarly. The definitions of the empirical copula given in (\ref{empcop}) and (\ref{empcopD}) are attributed to \citet{R76} and \citet{D79} respectively. They differ by at most \(2n^{-1}\) almost surely (a.s.); see Lemma \ref{frwlem} in the Appendix. The R\"{u}schendorf empirical copula is more convenient for computation, while the Deheuvels empirical copula is more convenient to analyze using the delta method.

The asymptotic validity of our proposed testing procedures hinges on weak convergence of the empirical copula process \(\mathbb C_n=\sqrt{n}\left(C_n-C\right)\) in the space \(\ell^\infty([0,1]^2)\) of bounded real valued functions on the unit square equipped with the uniform metric. Such weak convergence is satisfied under the following condition of \citet{S12}.
\begin{definition}
	A bivariate copula \(C\) is said to be regular if the partial derivatives \(\dot{C}_1\) and \(\dot{C}_2\) exist and are continuous everywhere on \((0,1)\times[0,1]\) and \([0,1]\times(0,1)\) respectively.
\end{definition}
We extend the definition of \(\dot{C}_1\) to \([0,1]^2\) by setting
\begin{equation*}
\dot{C}_1(u,v)=\begin{cases}\limsup_{\delta\downarrow0}\delta^{-1}C(\delta,v)&\text{for }u=0,\\\limsup_{\delta\downarrow0}\delta^{-1}(v-C(1-\delta,v))&\text{for }u=1,\end{cases}
\end{equation*}
and similarly for \(\dot{C}_2\). When \(C\) is regular, \citet{S12} has shown that the empirical copula process satisfies \(\mathbb C_n\rightsquigarrow\mathbb C\) in \(\ell^\infty([0,1]^2)\), where \(\rightsquigarrow\) denotes Hoffmann-J\o rgensen convergence in some metric space \citep[pp.\ 107--108]{K08}. The limit \(\mathbb C\) can be written as
\begin{equation}\label{cdef1}
\mathbb C(u,v)=\mathbb B(u,v)-\dot{C}_1(u,v)\mathbb B(u,1)-\dot{C}_2(u,v)\mathbb B(1,v),
\end{equation}
with \(\mathbb B\) a centered Gaussian random element of \(\ell^\infty([0,1]^2)\) with continuous sample paths and covariance kernel
\begin{equation}\label{cdef2}
\cov\left(\mathbb B(u,v),\mathbb B(u',v')\right)=C(u\wedge u',v\wedge v')-C(u,v)C(u',v').
\end{equation}
An alternative demonstration of the weak convergence \(\mathbb C_n\rightsquigarrow\mathbb C\) based on the delta method was given by \citet{B11} and by \citet{BV13}. Earlier work by \citet{FRW04} accomplished the same thing under a condition somewhat stronger than regularity.

\subsection{Exchangeability test statistics}\label{ssecrefstat}

For any function \(\theta\in\ell^\infty([0,1]^2)\) we define \(\theta^\top\in\ell^\infty([0,1]^2)\) by \(\theta^\top(u,v)=\theta(v,u)\). The null hypothesis that \(C\) is exchangeable may then be written as \(C=C^\top\). Statistics for testing this null may be constructed from the difference \(C_n-C_n^\top\). Following \citet{GNQ12} we consider the three statistics
\begin{align}
R_n&=n\int_0^1\int_0^1\left(C_n(u,v)-C_n(v,u)\right)^2\mathrm{d}u\mathrm{d}v,\label{rref}\\
S_n&=n\int_0^1\int_0^1\left(C_n(u,v)-C_n(v,u)\right)^2C_n(\mathrm{d}u,\mathrm{d}v),\label{sref}\\
T_n&=\sqrt{n}\sup_{(u,v)\in[0,1]^2}\left\vert C_n(u,v)-C_n(v,u)\right\vert.\label{tref}
\end{align}
Equivalent expressions for \(R_n\), \(S_n\) and \(T_n\) more amenable to exact computation, with the integrals replaced by sums and the supremum replaced by a maximum over a finite set, have been given by \citet[Prop.\ 1]{GNQ12}. When \(C\) is regular and exchangeable, \citet[Prop.\ 3]{GNQ12} have shown that
\begin{align}
R_n&\convl\int_0^1\int_0^1\left(\mathbb{C}(u,v)-\mathbb{C}(v,u)\right)^2\mathrm{d}u\mathrm{d}v,\label{rreflim}\\
S_n&\convl\int_0^1\int_0^1\left(\mathbb{C}(u,v)-\mathbb{C}(v,u)\right)^2C(\mathrm{d}u,\mathrm{d}v),\label{sreflim}\\
T_n&\convl\sup_{(u,v)\in[0,1]^2}\left\vert\mathbb{C}(u,v)-\mathbb{C}(v,u)\right\vert.\label{treflim}
\end{align}

\subsection{Radial symmetry test statistics}\label{ssecradstat}

For any function \(\theta\in\ell^\infty([0,1]^2)\), we define \(\theta^\srv,\theta^\rot\in\ell^\infty([0,1]^2)\) by
\begin{align*}
\theta^\srv(u,v)&=u+v-1+\theta(1-u,1-v),\\
\theta^\rot(u,v)&=\theta(1-u,1-v).
\end{align*}
The null hypothesis that \(C\) is radially symmetric may be written as \(C=C^\srv\). 
Denote the empirical cdf of rotated rank pairs \((1-U_{ni},1-V_{ni})\) by
\begin{equation}\label{empsrvcop}
D_n(u,v)=\frac{1}{n}\sum_{i=1}^n\mathbbm{1}\left(1-U_{ni}\leq u,1-V_{ni}\leq v\right),\quad(u,v)\in[0,1]^2.
\end{equation}
The cdf \(D_n\) differs from \(C_n^\mathrm{s}\) by no more than \(4n^{-1}\) a.s.; see Lemma \ref{gnlem} in the Appendix. This motivates the use of the following statistics for testing radial symmetry:
\begin{align}
R_n'&=n\int_0^1\int_0^1\left(C_n(u,v)-D_n(u,v)\right)^2\mathrm{d}u\mathrm{d}v,\label{rrad}\\
S_n'&=n\int_0^1\int_0^1\left(C_n(u,v)-D_n(u,v)\right)^2C_n(\mathrm{d}u,\mathrm{d}v),\label{srad}\\
T_n'&=\sqrt{n}\sup_{(u,v)\in[0,1]^2}\left\vert C_n(u,v)-D_n(u,v)\right\vert.\label{trad}
\end{align}
The statistics \(R_n'\), \(S_n'\) and \(T_n'\) are the same as those of \citet{GN14}, except that those authors defined \(U_{ni}\) and \(V_{ni}\) in a way that differs from us by an asymptotically negligible factor of \(n/(n+1)\). \citet[p.\ 1109]{GN14} provided equivalent expressions for these statistics more amenable to computation, with integrals replaced by sums and the supremum replaced by a maximum over a finite set. When \(C\) is regular and radially symmetric, results of \citet{BC12} imply that
\begin{align}
R_n'&\convl\int_0^1\int_0^1\left(\mathbb{C}(u,v)-\mathbb{C}^\rot(u,v)\right)^2\mathrm{d}u\mathrm{d}v,\label{rradlim}\\
S_n'&\convl\int_0^1\int_0^1\left(\mathbb{C}(u,v)-\mathbb{C}^\rot(u,v)\right)^2C(\mathrm{d}u,\mathrm{d}v),\label{sradlim}\\
T_n'&\convl\sup_{(u,v)\in[0,1]^2}\left\vert\mathbb{C}(u,v)-\mathbb{C}^\rot(u,v)\right\vert.\label{tradlim}
\end{align}

\section{Feasible randomization procedure}\label{secresample}

We require critical values with which to compare the test statistics defined in (\ref{rref})--(\ref{tref}) and (\ref{rrad})--(\ref{trad}). These should approximate the relevant quantiles of the null limit distributions given in (\ref{rreflim})-(\ref{treflim}) and (\ref{rradlim})-(\ref{tradlim}). The approach taken by \citet{GNQ12} and \citet{GN14} was to use the multiplier bootstrap of \citet{RS09} and \citet{BD10} to generate bootstrap versions of the empirical copula process \(\mathbb C_n\), and thereby approximate the null limit distribution of the relevant statistic. We will propose a different resampling scheme similar to randomization tests, narrowly tailored to the symmetry testing problem. It delivers improved small sample performance in simulations reported in Section \ref{secsim}.

Suppose for a moment that \(F\) and \(G\) are known, so that we observe \(n\) iid pairs \((U_i,V_i)=(F(X_i),G(Y_i))\) whose common bivariate cdf is the copula \(C\).
A randomization test of a hypothesis about \(C\) may be possible if there is a finite group \(\mathbf G\) of transformations from \([0,1]^{2}\) to itself such that, when the hypothesis is satisfied, \(g(U,V)\eql(U,V)\) for all \(g\in\mathbf G\). For testing bivariate exchangeability and radial symmetry it is enough to consider a group of two transformations \(\mathbf{G}=\{\pi^0,\pi^1\}\). The transformations \(\pi^0:[0,1]^2\to[0,1]^2\) and \(\pi^1:[0,1]^2\to[0,1]^2\) are given by
\begin{equation}\label{expi}
\pi^0(u,v)=(u,v)\quad\text{and}\quad\pi^1(u,v)=(v,u)
\end{equation}
when our hypothesis is exchangeability, and by
\begin{equation}\label{rspi}
\pi^0(u,v)=(u,v)\quad\text{and}\quad\pi^1(u,v)=(1-u,1-v)
\end{equation}
when our hypothesis is radial symmetry. It is easy to check that either choice of \(\mathbf G\) is a group under the operation of composition. Using \(\mathbf G\) we can construct a second group \(\mathbf G_n\) of \(2^n\) transformations from \(([0,1]^2)^n\) to itself by setting \(\mathbf G_n=\{g_\tau:\tau=(\tau_1,\ldots,\tau_n)\in\{0,1\}^n\}\), where
\begin{equation}\label{gn}
g_\tau\left((u_1,v_1),\ldots,(u_n,v_n)\right)=\left(\pi^{\tau_1}(u_1,v_1),\ldots,\pi^{\tau_n}(u_n,v_n)\right).
\end{equation}
Again, it is easy to check that \(\mathbf G_n\) is a group under the operation of composition. Moreover, when the relevant symmetry hypothesis is satisfied, our sample of pairs satisfies
\begin{equation}\label{gninvar}
g\left((U_1,V_1),\ldots,(U_n,V_n)\right)\eql\left((U_1,V_1),\ldots,(U_n,V_n)\right)\quad\text{for all }g\in\mathbf G_n.
\end{equation}

Given an arbitrary test statistic \(W_n=W_n((U_1,V_1),\ldots,(U_n,V_n))\), property \eqref{gninvar} can be used to justify the construction of an exact level \(\alpha\) randomization test of our symmetry hypothesis. The procedure, as described by \citet{R90}, is as follows.
\begin{procedure}[Infeasible randomization test]\label{rtest}\leavevmode
\begin{enumerate}
	\item Compute \(W_n^\tau=W_n(g_\tau((U_1,V_1),\ldots,(U_n,V_n)))\) for each \(g_\tau\in\mathbf G_n\). Denote by \(W_n^{(1)}\leq W_n^{(2)}\leq\cdots\leq W_n^{(2^n)}\) the ordered values of these statistics.
	\item Let \(k=2^n-\lfloor2^n\alpha\rfloor\), where \(\lfloor\cdot\rfloor\) rounds down to the nearest integer. Determine the number \(M^+\) of \(W_n^\tau\)'s that are strictly greater than \(W_n^{(k)}\), and the number \(M^0\) that are equal to \(W_n^{(k)}\).
	\item Reject the null if \(W_n>W_n^{(k)}\). Reject the null with probability \((2^n\alpha-M^+)/M^0\) if \(W_n=W_n^{(k)}\). Do not reject the null if \(W_n<W_n^{(k)}\).
\end{enumerate}
\end{procedure}
It can be shown that, when the null hypothesis of symmetry is true, this procedure leads us to reject it with probability exactly equal to \(\alpha\). See \citet{R89,R90} and references therein for further details on randomization tests.

Unfortunately a randomization test of the kind just described is not feasible for our hypothesis testing problem because, since \(F\) and \(G\) are not known, we do not observe pairs \((U_i,V_i)\) drawn from \(C\). In their place, we observe rank pairs \((U_{ni},V_{ni})\) based on preliminary estimates of the margins \(F\) and \(G\), but these rank pairs do not satisfy the exact group invariance property \eqref{gninvar}. A naive application of Procedure \ref{rtest} to these rank pairs will not achieve exact size, and cannot be expected to achieve correct size asymptotically because the effect of estimating the margins upon the distribution of the test statistic is not properly accounted for. We instead propose a feasible randomization procedure which accounts for the estimation of margins and will be shown to deliver asymptotically valid inference.

Let \(W_n=W_n((U_{n1},V_{n1}),\ldots,(U_{nn},V_{nn}))\) be the statistic of interest. Our procedure for computing critical values for \(W_n\) is as follows.
\begin{procedure}[Feasible randomization test]\label{qrtest}\leavevmode
\begin{enumerate}
	\item Select a transform \(g_\tau\in\mathbf G_n\). Set
	\begin{equation*}
	\left((U_{n1}^\tau,V_{n1}^\tau),\ldots,(U_{nn}^\tau,V_{nn}^\tau)\right)=g_\tau\left((U_{n1},V_{n1}),\ldots,(U_{nn},V_{nn})\right).
	\end{equation*}
	\item Draw $n$ iid random variables $\eta_1,\ldots,\eta_n$ from the uniform distribution on $(0,1)$. For $i=1,\ldots,n$, set
	\begin{equation*}
	\check{U}_{ni}^\tau=U_{ni}^\tau-n^{-1}\eta_i,\quad\check{V}_{ni}^\tau=V_{ni}^\tau-n^{-1}\eta_i.
	\end{equation*}
	\item For \(i=1,\ldots,n\) compute
	\begin{equation*}
	\tilde{U}_{ni}^\tau=\frac{1}{n}\sum_{j=1}^n\mathbbm{1}\left(\check{U}_{nj}^\tau\leq \check{U}_{ni}^\tau\right),\quad\tilde{V}_{ni}^\tau=\frac{1}{n}\sum_{j=1}^n\mathbbm{1}\left(\check{V}_{nj}^\tau\leq \check{V}_{ni}^\tau\right).
	\end{equation*}
	\item Compute \(W_n^\tau=W_n((\tilde{U}_{n1}^\tau,\tilde{V}_{n1}^\tau),\ldots,(\tilde{U}_{nn}^\tau,\tilde{V}_{nn}^\tau))\).
	\item Repeat Steps 1--4 to compute \(W_n^\tau\) for each \(g_\tau\in\mathbf G_n\). Denote by \(W_n^{(1)}\leq W_n^{(2)}\leq\cdots\leq W_n^{(2^n)}\) the ordered values of these statistics. Let \(k=2^n-\lfloor2^n\alpha\rfloor\), where \(\lfloor\cdot\rfloor\) rounds down to the nearest integer. Determine the number \(M^+\) of \(W_n^\tau\)'s that are strictly greater than \(W_n^{(k)}\), and the number \(M^0\) that are equal to \(W_n^{(k)}\).
	\item Reject the null if \(W_n>W_n^{(k)}\). Reject the null with probability \((2^n\alpha-M^+)/M^0\) if \(W_n=W_n^{(k)}\). Do not reject the null if \(W_n<W_n^{(k)}\).
\end{enumerate}
\end{procedure}

Steps 2 and 3 of Procedure \ref{qrtest} are what distinguish it from a naive application of Procedure \ref{rtest} to the rank pairs \((U_{ni},V_{ni})\). Step 2 involves applying small random perturbations \(n^{-1}\eta_i\) to the transformed rank pairs \((U^\tau_{ni},V^\tau_{ni})\). Since the randomized statistic \(W_n^\tau\) depends only on the rank pairs of the transformed rank pairs \((U_{ni}^\tau,V_{ni}^\tau)\), and the perturbations \(n^{-1}\eta_i\) are smaller than the gap \(n^{-1}\) between consecutive ranks, the sole effect of the perturbations is to randomly break ties in ranks. This effect is asymptotically negligible (see Lemma \ref{lemeta} below) but we have found that it improves the performance of our procedure with very small sample sizes. It is immaterial that \(\eta_i\) is uniformly distributed; we could equivalently draw each $\eta_i$ from any continuous distribution with support contained in \((0,1)\).

The more important distinction between Procedures \ref{rtest} and \ref{qrtest} is Step 3 of Procedure \ref{qrtest}, in which \((\check{U}_{ni}^\tau,\check{V}_{ni}^\tau)\) is transformed to \((\tilde{U}^\tau_{ni},\tilde{V}^\tau_{ni})\). If we were to drop Step 3 in Procedure \ref{qrtest} and naively use \((U_{ni}^\tau,V_{ni}^\tau)\) or \((\check{U}_{ni}^\tau,\check{V}_{ni}^\tau)\) in place of \((\tilde{U}_{ni}^\tau,\tilde{V}_{ni}^\tau)\) in Step 4, our method of constructing critical values would not deliver asymptotically valid inference. We illustrate this claim in Figure \ref{naiveapprox}. In panel (a) we plot smoothed histograms of the simulated sampling distribution of the exchangeability test statistic \(S_n\), of the distribution of randomized versions of \(S_n\) obtained by naively applying Procedure \ref{rtest} to the rank pairs \((U_{ni},V_{ni})\), and of the distribution of proper randomized versions of \(S_n\) obtained by applying Procedure \ref{qrtest} to the rank pairs \((U_{ni},V_{ni})\). It is clear that the naive application of Procedure \ref{rtest} to the rank pairs \((U_{ni},V_{ni})\) provides a very poor approximation to the sampling distribution of \(S_n\), whereas the approximation provided by Procedure \ref{qrtest} is very good. Moreover, we found that the latter approximation does not change appreciably for this sample size if the perturbation in Step 2 of Procedure \ref{qrtest} is omitted. Panel (b) shows qualitatively similar results for the radial symmetry test statistic \(S_n'\). We obtained qualitatively similar results for the exchangeability test statistics \(R_n\) and \(T_n\), and the radial symmetry test statistics \(R_n'\) and \(T_n'\), hence we do not display them here.

\begin{figure}[t!]
	\begin{center}	
		\begin{tikzpicture}[scale=4]
		\tikzstyle{vertex}=[font=\small,circle,draw,fill=yellow!20]
		\tikzstyle{edge} = [font=\scriptsize,draw,thick,-]
		\draw[black, thick] (0,0) -- (0,1);
		\draw[black, thick, ->] (0,0) -- (1.1,0);
		\draw (0,0.3pt) -- (0,-1pt)
		node[anchor=north,font=\scriptsize] {$0$};
		\draw (.2,0.3pt) -- (.2,-1pt)
		node[anchor=north,font=\scriptsize] {$.1$};
		\draw (.4,0.3pt) -- (.4,-1pt)
		node[anchor=north,font=\scriptsize] {$.2$};
		\draw (.6,0.3pt) -- (.6,-1pt)
		node[anchor=north,font=\scriptsize] {$.3$};
		\draw (.8,0.3pt) -- (.8,-1pt)
		node[anchor=north,font=\scriptsize] {$.4$};
		\draw (1,0.3pt) -- (1,-1pt)
		node[anchor=north,font=\scriptsize] {$.5$};
		\draw (0.3pt,0) -- (-0.8pt,0);
		\draw (0.3pt,{2/7}) -- (-0.8pt,{2/7});
		\draw (0.3pt,{4/7}) -- (-0.8pt,{4/7});
		\draw (0.3pt,{6/7}) -- (-0.8pt,{6/7});
		\node[left,font=\scriptsize] at (0,0) {$0$};
		\node[left,font=\scriptsize] at (0,{2/7}) {$10$};
		\node[left,font=\scriptsize] at (0,{4/7}) {$20$};
		\node[left,font=\scriptsize] at (0,{6/7}) {$30$};
		\draw[red,thick] plot[smooth] file {EXCtrue.txt};	
		\draw[green!60!black,dashed,thick] plot[smooth] file {EXCfeasible.txt};	
		\draw[blue,dotted,thick] plot[smooth] file {EXCnaive.txt};		
		\node[font=\small] at (0.5,-0.4) {(a) Exchangeability statistic ($S_n$)};
		\end{tikzpicture}
		\,\,\,\,
		\begin{tikzpicture}[scale=4]
		\tikzstyle{vertex}=[font=\small,circle,draw,fill=yellow!20]
		\tikzstyle{edge} = [font=\scriptsize,draw,thick,-]
		\draw[black, thick] (0,0) -- (0,1);
		\draw[black, thick, ->] (0,0) -- (1.1,0);
		\draw (0,0.3pt) -- (0,-1pt)
		node[anchor=north,font=\scriptsize] {$0$};
		\draw (.2,0.3pt) -- (.2,-1pt)
		node[anchor=north,font=\scriptsize] {$.1$};
		\draw (.4,0.3pt) -- (.4,-1pt)
		node[anchor=north,font=\scriptsize] {$.2$};
		\draw (.6,0.3pt) -- (.6,-1pt)
		node[anchor=north,font=\scriptsize] {$.3$};
		\draw (.8,0.3pt) -- (.8,-1pt)
		node[anchor=north,font=\scriptsize] {$.4$};
		\draw (1,0.3pt) -- (1,-1pt)
		node[anchor=north,font=\scriptsize] {$.5$};
		\draw (0.3pt,0) -- (-0.8pt,0);
		\draw (0.3pt,{2/7}) -- (-0.8pt,{2/7});
		\draw (0.3pt,{4/7}) -- (-0.8pt,{4/7});
		\draw (0.3pt,{6/7}) -- (-0.8pt,{6/7});
		\node[left,font=\scriptsize] at (0,0) {$0$};
		\node[left,font=\scriptsize] at (0,{2/7}) {$10$};
		\node[left,font=\scriptsize] at (0,{4/7}) {$20$};
		\node[left,font=\scriptsize] at (0,{6/7}) {$30$};
		\draw[red,thick] plot[smooth] file {RADtrue.txt};	
		\draw[green!60!black,dashed,thick] plot[smooth] file {RADfeasible.txt};	
		\draw[blue,dotted,thick] plot[smooth] file {RADnaive.txt};			
		\node[font=\small] at (0.5,-0.4) {(b) Radial symmetry statistic ($S_n'$)};
		\end{tikzpicture}
	\end{center}
	\caption{Smoothed histograms of the sampling distributions of \(S_n\) and \(S_n'\) (\textcolor{red}{solid} lines), of the approximating distributions obtained by naively applying Procedure \ref{rtest} to the rank pairs \((U_{ni},V_{ni})\) (\textcolor{blue}{dotted} lines), and of the approximating distributions obtained by applying Procedure \ref{qrtest} to the rank pairs \((U_{ni},V_{ni})\) (\textcolor{green!60!black}{dashed} lines). The approximation is done for a sample size of \(n=100\) and with \(C\) chosen to be the product copula, so that the null hypothesis of symmetry is satisfied.}
	\label{naiveapprox}
\end{figure}
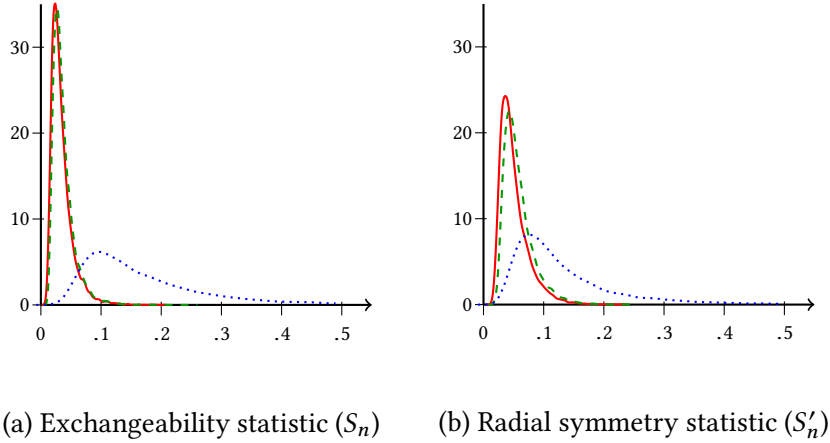

To study the asymptotic properties of tests based on Procedure \ref{qrtest} it will be useful to consider the behavior of \(W_n^\tau\) when \(\tau\) is drawn randomly from \(\mathbf G_n\). Suppose we were to choose as \(\tau\) an \(n\)-tuple of Bernoulli random variables each taking the values zero and one with equal probabilities, independent of one another and of the data. The corresponding transform \(g_\tau\) would then be a random draw from the uniform distribution over \(\mathbf G_n\), and the critical value \(W_n^{(k)}\) computed in step 4 of Procedure \ref{qrtest} would be the conditional \((1-\alpha)\)-quantile of \(W_n^\tau\) given the data. That is, \(W_n^{(k)}=Q_n(1-\alpha)\), where
\begin{equation}\label{Qn}
Q_n(u)=\inf\left\{x\in\mathbf R:\pr(W^\tau_n\leq x\,|\,(X_1,Y_1),\ldots,(X_n,Y_n))\geq u\right\}.
\end{equation}
We will show in Section \ref{secasym} that, when the null of symmetry is satisfied, the conditional law of \(W_n^\tau\) given the data asymptotically approximates the unconditional law of the test statistic \(W_n\). Consequently, our critical value \(W_n^{(k)}\) approximates the \((1-\alpha)\)-quantile of the law of \(W_n\) under the null, as desired.

In practice it may be computationally burdensome to evaluate all \(2^n\) statistics \(W_n^\tau\) corresponding to the \(2^n\) transforms \(g_\tau\in\mathbf G_n\), even with modest sample sizes. Instead of repeating steps 1--4 of Procedure \ref{qrtest} \(2^n\) times, we may repeat them some large number of times \(N\). Each time we select a transform \(g_\tau\in\mathbf G_n\) in Step 1, we should do so at random and with replacement, choosing each possible transform with equal probabilities. This amounts to generating \(\tau\) as an independent \(n\)-tuple of Bernoulli random variables, each taking the values zero and one with equal probabilities. In Step 5 of Procedure \ref{qrtest} we then set our critical value equal to the \(k^\text{th}\) smallest of the \(N\) computed statistics, where \(k=N-\lfloor N\alpha\rfloor\). So long as \(N\) is large, our critical value should be close to the conditional \((1-\alpha)\)-quantile of \(W_n^\tau\) given the data.

Step 6 in Procedure \ref{qrtest} entails randomizing between rejection and nonrejection when the test statistic $W_n$ and critical value $W_n^{(k)}$ are equal. Applied researchers may be reluctant to randomize in this way, and instead prefer to only reject the null when the test statistic strictly exceeds the critical value. In numerical simulations discussed in Section \ref{secsim} we found that the effect of always failing to reject, rather than randomizing, when the test statistic and critical value are equal is negligible when using the exchangeability test statistics $R_n$ and $S_n$ and the radial symmetry test statistics $R_n'$ and $S_n'$, even at sample sizes as small as $n=30$. On the other hand, choosing not to randomize can have a noticeable impact when using the statistics $T_n$ and $T_n'$. This may be because, holding $W_n$ fixed, the range of values taken by $W_n^\tau$ is much coarser when we construct $W_n$ using the uniform norm, leading to $W_n$ being exactly equal to its critical value with greater frequency. Although this discrepancy diminishes at larger sample sizes, applied researchers unwilling to randomize may prefer not to use the statistics $T_n$ and $T_n'$ in small samples.

\section{Asymptotic properties}\label{secasym}

\subsection{Conditional weak convergence of randomized statistics}

Though the procedure we have proposed may be viewed as a feasible version of a randomization test, we will study it using techniques most often used to demonstrate the asymptotic validity of bootstrap tests. Such demonstrations generally hinge upon the conditional law of a bootstrapped statistic converging in a suitable sense to a target distribution. The conditional law we refer to here is the law obtained by holding the data fixed and allowing the random weights used to generate the bootstrapped statistic to vary. The analysis of our procedure will be similar, in that it hinges upon suitable convergence of the conditional law of the randomized statistic \(W_n^\tau\). However in this case the source of random variation is not a collection of bootstrap weights, but rather a random \(n\)-tuple \(\tau=(\tau_1,\ldots,\tau_n)\) that indexes transforms in the group \(\mathbf G_n\).

Throughout this section, \(\tau\) represents an \(n\)-tuple of independent Bernoulli random variables each taking the values zero and one with equal probabilities, jointly independent of the data \((X_1,Y_1),\ldots,(X_n,Y_n)\). For \(n\in\mathbf N\), \(\xi^\tau_n\) is an element of a metric space \(\mathbf D\) depending on the data and \(\tau\), and \(\xi_n\) is an element of \(\mathbf D\) depending on the data but not on \(\tau\). The following notion of convergence is from \citet[pp.\ 19--20]{K08}.
\begin{definition}\label{cwcdef}
	If \(\xi\) is a tight random element of \(\mathbf D\) then we say that \(\xi^\tau_n\) weakly converges to \(\xi\) conditional on the data in probability, and write \(\xi^\tau_n\cwconv\xi\), if
	\begin{enumerate}
		\item \(\sup_{f\in\mathrm{BL}_1(\mathbf D)}\vert\ex_\tau f(\xi^\tau_n)-\ex f(\xi)\vert\to0\) in outer probability and
		\item \(\ex_\tau f(\xi^\tau_n)^\ast-\ex_\tau f(\xi^\tau_n)_\ast\to0\) in probability for every \(f\in\mathrm{BL}_1(\mathbf D)\),
	\end{enumerate}
	where \(\mathrm{BL}_1(\mathbf D)\) is the set of real Lipschitz functions on \(\mathbf D\) with level and Lipschitz constant bounded by one, \(\ex_\tau\) is expectation over \(\tau\) holding the data fixed, and \(f(\xi^\tau_n)^\ast\) and \(f(\xi^\tau_n)_\ast\) are the minimal measurable majorant and maximal measurable minorant of \(f(\xi^\tau_n)\) with respect to the data and random index jointly.
\end{definition}
The second condition in Definition \ref{cwcdef}, which asserts a form of conditional asymptotic measurability, allows us to handle cases where \(\xi^\tau_n\) is not Borel measurable. While this level of generality can be useful in other contexts, whenever the symbols \(\rightsquigarrow\) or \(\cwconv\) are used in this paper, each member of the convergent sequence will in fact be Borel measurable. This is true even in the nonseparable space \(\mathbf D=\ell^\infty([0,1]^2)\). To see why, note that the empirical copula \(C_n\) is uniquely determined by the \(n^2\) coordinate projections \(C_n(i/n,j/n)\), \(i,j=1,\ldots,n\), and that each of these projections can only take the values \(0,n^{-1},2n^{-1},\ldots,1\). Thus \(C_n\) can take only finitely many values in \(\ell^\infty([0,1]^2)\), and since the projections \(C_n(i/n,j/n)\) are random variables, \(C_n\) must be a simple map into \(\ell^\infty([0,1]^2)\). Similarly, \(D_n\) is a simple map into \(\ell^\infty([0,1]^2)\). The weakly convergent or conditionally weakly convergent sequences we consider in this paper are all simple, hence Borel measurable.

Theorems \ref{exthm} and \ref{rsthm} below establish that the randomized statistic \(W_n^\tau\) satisfies \(W_n^\tau\cwconv W\), where the law of \(W\) coincides with the weak limit of \(W_n\) when the null of symmetry is satisfied. From this conditional weak convergence and Lemma 10.11 of \citet{K08} it follows that
\begin{equation}\label{fslem0}
\pr(W_n^\tau\leq c\,|\,(X_1,Y_1),\ldots,(X_n,Y_n))\to\pr(W\leq c)
\end{equation}
in probability for all continuity points \(c\) of the cdf of \(W\). Consequently, in view of \eqref{Qn}, our feasible randomization tests control size asymptotically and are consistent against arbitrary violations of symmetry.

Recent work on the asymptotic properties of randomization tests by \citet{CR13,CR16a,CR16b} establishes conditional weak convergence of their randomization distributions by verifying an equivalent condition of \citet{H52}. Let $\tilde{W}_n(\cdot)$ be the map from randomized rank-pairs $(U_{n1}^\tau,V_{n1}^\tau),\ldots,(U_{nn}^\tau,V_{nn}^\tau)$ to randomized test statistic $W_n^\tau$, so that
\begin{equation}
W_n^\tau=\tilde{W}_n(g_\tau((U_{n1},V_{n1}),\ldots,(U_{nn},V_{nn}))).
\end{equation}
The conditional probability in \eqref{fslem0} is equal to
\begin{equation}
\frac{1}{2^n}\sum_{g\in\mathbf G_n}\mathbbm{1}(\tilde{W}_n(g((U_{n1},V_{n1}),\ldots,(U_{nn},V_{nn})))\leq c),
\end{equation}
and corresponds to the randomization distribution defined in Equation (5.1) of \citet{CR13}. Theorem 5.1 of \citet{CR13} therefore implies that \eqref{fslem0} is satisfied if and only if
\begin{equation}\label{hoeffding}
(W_n^\tau,W_n^{\tau'})\rightsquigarrow(W,W'),
\end{equation}
where \(W'\) and \(\tau'\) are independent copies of \(W\) and \(\tau\) respectively. Condition \eqref{hoeffding} is called Hoeffding's condition. We do not make use of Hoeffding's condition and instead verify \eqref{fslem0} directly, which we find to be more convenient.

The proofs of our main results rely on repeated applications of conditional versions of two fundamental results on weak convergence: the continuous mapping theorem and the delta method. We will state these here for convenience, and also because we require a version of the conditional delta method that is somewhat different to the usual statement. The following statement of the conditional continuous mapping theorem is Theorem 10.8 of \citet{K08}. We have dropped Kosorok's measurability condition on \(\tau\mapsto\xi_n^\tau\) because \(\tau\) takes only finitely many values.

\begin{thm}[Conditional continuous mapping theorem]
	Let \(\phi:\mathbf D\to\mathbf E\) be continuous at all points in \(\mathbf D_0\subset\mathbf D\), where \(\mathbf D\) and \(\mathbf E\) are Banach spaces and \(\mathbf D_0\) is closed. If \(\xi^\tau_n\cwconv\xi\) in \(\mathbf D\), where \(\xi\) is tight and concentrates on \(\mathbf D_0\), then \(\phi(\xi^\tau_n)\cwconv\phi(\xi)\) in \(\mathbf E\).
\end{thm}

The following statement of the conditional delta method is a version of Theorem 12.1 of \citet{K08}, often referred to as the delta method for the bootstrap. It is unusual in that we allow the weak limits \(\mathbb X_1\) and \(\mathbb X_2\) to differ (and not merely by a scalar multiple \(c\)). We have again dropped the measurability condition on \(\tau\mapsto\xi_n^\tau\).

\begin{thm}[Conditional delta method]\label{cdm}
	Let \(\mathbf D\) and \(\mathbf E\) be Banach spaces and let \(\phi:\mathbf D_\phi\subset\mathbf D\to\mathbf E\) be Hadamard differentiable at \(\mu\in\mathbf D_\phi\) tangentially to \(\mathbf D_0\subset\mathbf D\), with derivative \(\phi'_\mu\). Let \(\xi_n\) and \(\xi_n^\tau\) take values in \(\mathbf D_\phi\), and suppose that \(\sqrt{n}(\xi_n-\mu)\rightsquigarrow\mathbb X_1\) and \(\sqrt{n}(\xi_n^\tau-\xi_n)\cwconv\mathbb X_2\) in \(\mathbf D\), where \(\mathbb X_1\) and \(\mathbb X_2\) are tight and take values in \(\mathbf D_0\). Then \(\sqrt{n}(\phi(\xi_n^\tau)-\phi(\xi_n))\cwconv\phi'_\mu(\mathbb X_2)\) in \(\mathbf E\).
\end{thm}
To obtain Theorem \ref{cdm} we set \(c=1\) in the proof of Theorem 12.1 of \citet{K08}, and when Kosorok introduces two independent copies \(\tilde{\mathbb X}_1\) and \(\tilde{\mathbb X}_2\) of \(\mathbb X\), we instead introduce two independent copies \(\tilde{\mathbb X}_1\) and \(\tilde{\mathbb X}_2\) of \(\mathbb X_1\) and \(\mathbb X_2\) respectively. It is crucial for us to allow the laws of \(\mathbb X_1\) and \(\mathbb X_2\) to differ. In the proofs of Lemmas \ref{exlem} and \ref{rslem} and Theorems \ref{exthm} and \ref{rsthm} below, we apply the conditional delta method even though the weak limits in \eqref{cwcexpf} and \eqref{uwcexpf}, in \eqref{product1} and \eqref{product2}, in \eqref{cwcrspf3} and \eqref{cwcrspf4}, and in \eqref{product1rs} and \eqref{product2rs}, differ.

Asymptotic justifications of bootstrap procedures typically appeal directly or indirectly to a multiplier central limit theorem. For instance, \citet{RS09} appeal directly to a multiplier central limit theorem to demonstrate the asymptotic validity of their procedure, whereas \citet{BD10} appeal to Theorem 2.6 of \citet{K08}, which is proved using a multiplier central limit theorem. This approach is less suitable for our problem, because the source of variation in the randomized statistic \(W_n^\tau\) is not a random \(n\)-tuple of independent bootstrap weights (i.e.\ multipliers), but rather a random \(n\)-tuple of transforms drawn independently from \(\mathbf G\). The following lemma plays the role of a multiplier central limit theorem in our analysis, in the sense that it shows how random variation in the draws from \(\mathbf G\) can provide a source of conditional weak convergence to a suitable Gaussian limit. It is proved in Section \ref{appx} by verifying the conditions of a functional central limit theorem of \citet{P90}.

\begin{lemma}\label{lemcwc}
	Let \(\pi^0\) and \(\pi^1\) be the pair of maps defined in either \eqref{expi} or \eqref{rspi}, and for \(i=1,\ldots,n\) and \((u,v)\in[0,1]^2\) let
	\begin{equation*}
	Z_{ni}(u,v)=\mathbbm{1}\left(\pi^{\tau_i}(U_{ni},V_{ni})\leq(u,v)\right)-\pr\left(\pi^{\tau_i}(U_{ni},V_{ni})\leq(u,v)\,\vert\,(U_{ni},V_{ni})\right).
	\end{equation*}
	Let \(\xi^\tau_n=n^{-1/2}\sum_{i=1}^nZ_{ni}\). Then \(\xi^\tau_n\cwconv\mathbb A\) in \(\ell^\infty([0,1]^2)\), where \(\mathbb A\) is centered and Gaussian with continuous sample paths. The covariance kernel of \(\mathbb A\) is given by
	\begin{align}\label{kerAex}
	\cov\left(\mathbb A(u,v),\mathbb A(u',v')\right))&=\frac{1}{4}C(u\wedge u',v\wedge v')+\frac{1}{4}C(v\wedge v',u\wedge u')\notag\\&\quad-\frac{1}{4}C(u\wedge v',v\wedge u')-\frac{1}{4}C(v\wedge u',u\wedge v')
	\end{align}
	if \(\pi^0\) and \(\pi^1\) are defined as in \eqref{expi}, or by
	\begin{align}\label{kerArs}
	\cov\left(\mathbb A(u,v),\mathbb A(u',v')\right))&=\frac{1}{4}C(u\wedge u',v\wedge v')+\frac{1}{4}C^\srv(u\wedge u',v\wedge v')\notag\\&\quad-\frac{1}{4}\mathrm{P}\left(1-u'<U\leq u,1-v'<V\leq v\right)\notag\\&\quad-\frac{1}{4}\mathrm{P}\left(1-u<U\leq u',1-v<V\leq v'\right)
	\end{align}
	if \(\pi^0\) and \(\pi^1\) are defined as in \eqref{rspi}.
\end{lemma}

\subsection{Asymptotic properties: exchangeability test}\label{exsasym}

In this subsection we let \(W_n=W_n((U_{n1},V_{n1}),\ldots,(U_{nn},V_{nn}))\) denote any of the statistics \(R_n\), \(S_n\) and \(T_n\), and let \(\pi^0\) and \(\pi^1\) be as defined in \eqref{expi}. Let \((\tilde{U}_{ni}^\tau,\tilde{V}_{ni}^\tau)\), \(i=1,\ldots,n\), be constructed as described in Steps 1--3 of Procedure \ref{qrtest}, with \(\tau\) an \(n\)-tuple of independent Bernoulli random variables each taking the values zero and one with equal probabilities. Let \(C_n^\tau\) be the random element of \(\ell^\infty([0,1]^2)\) given by
\begin{equation}\label{ctauex}
C_n^\tau(u,v)=\frac{1}{n}\sum_{i=1}^n\mathbbm{1}\left(\tilde{U}_{ni}^\tau\leq u,\tilde{V}_{ni}^\tau\leq v\right).
\end{equation}
The randomized statistic \(W_n^\tau=W_n((\tilde{U}^\tau_{n1},\tilde{V}^\tau_{n1}),\ldots,(\tilde{U}^\tau_{nn},\tilde{V}^\tau_{nn}))\) can be expressed in terms of \(C_n^\tau\): we simply substitute \(C_n^\tau\) for \(C_n\) in the relevant formula from \eqref{rref}--\eqref{tref}. The following lemma describes the conditional asymptotic behavior of \(C_n^\tau\).

\begin{lemma}\label{exlem}
	Suppose that \(C\) is regular. Then
	\begin{equation}\label{excwc2}
	\sqrt{n}\left(C_n^\tau-\frac{1}{2}(C_n+C_n^\top)\right)\cwconv\mathbb D,
	\end{equation}
	where \(\mathbb D\) can be written as
	\begin{equation}\label{exd}
	\mathbb D(u,v)=\mathbb A(u,v)-\frac{1}{2}(\dot{C}_1(u,v)+\dot{C}_2(v,u))\mathbb A(u,1)-\frac{1}{2}(\dot{C}_2(u,v)+\dot{C}_1(v,u))\mathbb A(1,v),
	\end{equation}
	and \(\mathbb A\) is a centered Gaussian random element of \(\ell^\infty([0,1]^2)\) with continuous sample paths and covariance kernel given by \eqref{kerAex}.
\end{lemma}

From Lemma \ref{exlem} and the conditional continuous mapping theorem we have
\begin{equation}\label{cwccmct}
\sqrt{n}\left(C_n^\tau-C_n^{\tau\top}\right)\cwconv\mathbb D-\mathbb D^\top.
\end{equation}
On the other hand, it is apparent from \eqref{rreflim}--\eqref{treflim} that it is the law of \(\mathbb C-\mathbb C^\top\) which we wish to approximate. In fact, Lemma \ref{exlem2} below establishes that, when the null of exchangeability is satisfied, \(\mathbb C-\mathbb C^\top\) and \(\mathbb D-\mathbb D^\top\) are equal in law. This leads us to our main result on our procedure for testing copula exchangeability.

\begin{thm}\label{exthm}
	Let \(W\) be the real valued random variable given by
	\begin{equation*}
	W=
	\begin{cases}
	\int_0^1\int_0^1\left(\mathbb{D}(u,v)-\mathbb{D}(v,u)\right)^2\mathrm{d}u\mathrm{d}v&\text{if }W_n=R_n\\
	\frac{1}{2}\int_0^1\int_0^1\left(\mathbb{D}(u,v)-\mathbb{D}(v,u)\right)^2(C(\mathrm{d}u,\mathrm{d}v)+C(\mathrm{d}v,\mathrm{d}u))&\text{if }W_n=S_n\\
	\sup_{(u,v)\in[0,1]^2}\left\vert\mathbb{D}(u,v)-\mathbb{D}(v,u)\right\vert&\text{if }W_n=T_n.
	\end{cases}
	\end{equation*}
	If \(C\) is regular then \(W_n^\tau\cwconv W\). If also \(C=C^\top\) then \(\mathbb C-\mathbb C^\top\eql\mathbb D-\mathbb D^\top\) and thus the law of \(W\) is the relevant weak limit from \eqref{rreflim}--\eqref{treflim}.
\end{thm}

We saw in Figure \ref{naiveapprox} that the naive application of Procedure \ref{rtest} to the pairs \((U_{ni},V_{ni})\) does not provide an acceptable approximation to the sampling distribution of our test statistic. To see why, let \(\tilde{C}_n^\tau\) be given by
\begin{equation}\label{Ftauex}
\tilde{C}_n^\tau(u,v)=n^{-1}\sum_{i=1}^n\mathbbm{1}(U_{ni}^\tau\leq u,V_{ni}^\tau\leq v),
\end{equation}
where \((U_{ni}^\tau,V_{ni}^\tau)=\pi^{\tau_i}(U_{ni},V_{ni})\) and \(\pi^0\) and \(\pi^1\) are defined as in \eqref{expi}. Naive application of Procedure \ref{rtest} to the rank pairs \((U_{ni},V_{ni})\) amounts to using \(\tilde{C}_n^\tau\) in place of \(C_n^\tau\) in our feasible randomization procedure. By applying the conditional continuous mapping theorem to the convergence \eqref{cwcexpf} established in the proof of Lemma \ref{exlem}, we find that
\begin{equation}\label{cwcfmft}
\sqrt{n}(\tilde{C}_n^\tau-\tilde{C}_n^{\tau\top})\cwconv\mathbb A-\mathbb A^\top.
\end{equation}
Comparing \eqref{cwccmct} and \eqref{cwcfmft} we see that using \(\tilde{C}_n^\tau\) in place of \(C_n^\tau\) leads to an invalid approximation of the weak limit \(\mathbb C-\mathbb C^\top\eql\mathbb D-\mathbb D^\top\) appearing in \eqref{rreflim}--\eqref{treflim}. The problem is that \(\mathbb A\) differs from \(\mathbb D\) by the sum of the latter two terms on the right hand side of equality \eqref{exd}. An analogous problem confounds the naive application of Procedure \ref{rtest} in the context of radial symmetry testing.

\subsection{Asymptotic properties: radial symmetry test}\label{rssasym}

In this subsection we let \(W_n=W_n((U_{n1},V_{n1}),\ldots,(U_{nn},V_{nn}))\) denote any of the statistics \(R_n'\), \(S_n'\) and \(T_n'\), and let \(\pi^0\) and \(\pi^1\) be as defined in \eqref{rspi}. We let \(C_n^\tau\) be defined in the same fashion as in the previous subsection, and let \(D_n^\tau\) be the random element of \(\ell^\infty([0,1]^2)\) given by
\begin{equation*}
D_n^\tau(u,v)=\frac{1}{n}\sum_{i=1}^n\mathbbm{1}\left(1-\tilde{U}_{ni}^\tau\leq u,1-\tilde{V}_{ni}^\tau\leq v\right).
\end{equation*}
The randomized statistic \(W_n^\tau=W_n((\tilde{U}^\tau_{n1},\tilde{V}^\tau_{n1}),\ldots,(\tilde{U}^\tau_{nn},\tilde{V}^\tau_{nn}))\) can be expressed in terms of \(C_n^\tau\) and \(D_n^\tau\): we simply substitute \(C_n^\tau\) for \(C_n\) and \(D_n^\tau\) for \(D_n\) in the relevant formula from \eqref{rrad}--\eqref{trad}. The following lemma, analogous to Lemma \ref{exlem}, describes the conditional asymptotic behavior of \(C_n^\tau\).

\begin{lemma}\label{rslem}
	Suppose that \(C\) is regular. Then
	\begin{equation}\label{rscwc2}
	\sqrt{n}\left(C_n^\tau-\frac{1}{2}(C_n+D_n)\right)\cwconv\mathbb D,
	\end{equation}
	where \(\mathbb D\) can be written as
	\begin{align}
	\mathbb D(u,v)&=\mathbb A(u,v)-\frac{1}{2}\left(\dot{C}_1(u,v)+1-\dot{C}_1(1-u,1-v)\right)\mathbb A(u,1)\notag\\&\quad-\frac{1}{2}\left(\dot{C}_2(u,v)+1-\dot{C}_2(1-u,1-v)\right)\mathbb A(1,v),
	\end{align}
	and \(\mathbb A\) is a centered Gaussian random element of \(\ell^\infty([0,1]^2)\) with continuous sample paths and covariance kernel given by \eqref{kerArs}.
\end{lemma}

Note that the random element \(\mathbb D\) appearing in the statement of Lemma \ref{rslem} is not the same as the random element \(\mathbb D\) appearing in the statement of Lemma \ref{exlem}.
It is apparent from \eqref{rradlim}--\eqref{tradlim} that the law of \(\mathbb C-\mathbb C^\rot\) (and also \(C\), if \(W_n=S_n'\)) determines the null limit distribution of \(W_n\). Lemma \ref{rslem2} below establishes that, when the null of radial symmetry is satisfied, \(\mathbb C-\mathbb C^\rot\) and \(\mathbb D-\mathbb D^\rot\) are equal in law. This leads us to our main result on our procedure for testing copula radial symmetry.

\begin{thm}\label{rsthm}
	Let \(W\) be the real valued random variable given by
	\begin{equation*}
	W=
	\begin{cases}
	\int_0^1\int_0^1\left(\mathbb{D}(u,v)-\mathbb{D}^\rot(u,v)\right)^2\mathrm{d}u\mathrm{d}v&\text{if }W_n=R_n'\\
	\frac{1}{2}\int_0^1\int_0^1\left(\mathbb{D}(u,v)-\mathbb{D}^\rot(u,v)\right)^2(C(\mathrm{d}u,\mathrm{d}v)+C^\srv(\mathrm{d}u,\mathrm{d}v))&\text{if }W_n=S_n'\\
	\sup_{(u,v)\in[0,1]^2}\left\vert\mathbb{D}(u,v)-\mathbb{D}^\rot(u,v)\right\vert&\text{if }W_n=T_n'.
	\end{cases}
	\end{equation*}
	If \(C\) is regular then \(W_n^\tau\cwconv W\). If also \(C=C^\srv\) then \(\mathbb C-\mathbb C^\rot\eql\mathbb D-\mathbb D^\rot\) and thus the law of \(W\) is the relevant weak limit from \eqref{rradlim}--\eqref{tradlim}.
\end{thm}

\section{Numerical simulations}\label{secsim}

\subsection{Size control}\label{secsize}

In Table \ref{nulltab} we report the results of numerical simulations used to investigate the null rejection rates obtained using our feasible randomization procedure (i.e., Procedure \ref{qrtest}) at small sample sizes. Rejection frequencies were computed over 10000 experimental replications for the sample sizes \(n=30\) and \(n=50\) at the nominal levels \(\alpha=0.05\) and \(\alpha=0.1\). Alongside the rejection rates for our feasible randomization procedure, we report rejection rates for the infeasible randomization procedure (Procedure \ref{rtest}) and for the multiplier bootstrap used by \citet{GNQ12} for testing exchangeability and by \citet{GN14} for testing radial symmetry. In each replication, critical values were computed using \(N=250\) random draws from \(\mathbf G_n\) for the randomization procedures and \(N=250\) repetitions for the multiplier bootstrap.

\begin{table}[p]
\vspace*{-3\baselineskip}
	{\footnotesize
		\begin{center}
			\begin{tabular}{c|c|ll|ccc|ccc|ccc}
				\hline\hline
				\multicolumn{1}{c}{}&\multicolumn{1}{c}{}&\multirow{2}{*}{Copula} & \multirow{2}{*}{$n$}&\multicolumn{3}{c}{\thead{Feasible\\ randomization}} &\multicolumn{3}{c}{\thead{Infeasible\\ randomization}} &\multicolumn{3}{c}{Multiplier bootstrap}\\
				\cline{5-13}
				\multicolumn{1}{c}{}&\multicolumn{1}{c}{}&&& $R_{n}$ & $S_{n}$ & $T_{n}$ & $R_{n}$ & $S_{n}$ & $T_{n}$ & $R_{n}$ & $S_{n}$ & $T_{n}$ \\
				\hline\hline
				\multirow{20}{*}{\rotatebox[origin=c]{90}{Exchangeability tests}}&\multirow{10}{*}{\rotatebox[origin=c]{90}{$\alpha=0.05$}}&\multirow{2}{*}{Gauss}&\multirow{1}{*}{30} &     0.064&     0.068&     0.067&     0.053&     0.051&     0.052&     0.014&     0.032&     0.003\\
				\cline{5-13}
				&&&\multirow{1}{*}{50}&     0.060&     0.064&     0.057&     0.048&     0.050&     0.051&     0.015&     0.035&     0.004\\
				\cline{3-13}
				&&\multirow{2}{*}{Clayton}&\multirow{1}{*}{30}&     0.069&     0.063&     0.062&     0.054&     0.053&     0.051&     0.016&     0.031&     0.002\\
				
				\cline{5-13}
				&&&\multirow{1}{*}{50}&     0.053&     0.056&     0.054&     0.053&     0.051&     0.053&     0.013&     0.025&     0.004\\
				\cline{3-13}
				&&\multirow{2}{*}{Gumbel}&\multirow{1}{*}{30}&     0.065&     0.062&     0.059&     0.056&     0.056&     0.054&     0.013&     0.028&     0.002\\
				\cline{5-13}
				&&&\multirow{1}{*}{50}&     0.060&     0.055&     0.060&     0.048&     0.047&     0.049&     0.012&     0.025&     0.005\\
				\cline{3-13}
				&&\multirow{2}{*}{Student}&\multirow{1}{*}{30}&     0.068&     0.059&     0.057&     0.050&     0.051&     0.050&     0.017&     0.031&     0.002\\
				\cline{5-13}
				&&&\multirow{1}{*}{50}&     0.058&     0.057&     0.050&     0.057&     0.058&     0.051&     0.018&     0.033&     0.002\\
				\cline{3-13}
				&&\multirow{2}{*}{Frank}&\multirow{1}{*}{30}&     0.065&     0.060&     0.061&     0.047&     0.050&     0.051&     0.012&     0.032&     0.003\\
				\cline{5-13}
				&&&\multirow{1}{*}{50}&     0.059&     0.058&     0.055&     0.052&     0.053&     0.050&     0.014&     0.026&     0.004\\
				\cline{2-13}
				&\multirow{10}{*}{\rotatebox[origin=c]{90}{$\alpha=0.1$}}&\multirow{2}{*}{Gauss}&\multirow{1}{*}{30} &     0.122&     0.129&     0.123&     0.104&     0.103&     0.103&     0.055&     0.071&     0.016\\
				\cline{5-13}
				&&&\multirow{1}{*}{50}&     0.113&     0.113&     0.109&     0.102&     0.097&     0.099&     0.053&     0.075&     0.016\\
				\cline{3-13}
				&&\multirow{2}{*}{Clayton}&\multirow{1}{*}{30}&     0.115&     0.117&     0.111&     0.103&     0.104&     0.103&     0.050&     0.067&     0.017\\
				\cline{5-13}
				&&&\multirow{1}{*}{50}&     0.106&     0.111&     0.109&     0.105&     0.100&     0.100&     0.049&     0.070&     0.016\\
				\cline{3-13}
				&&\multirow{2}{*}{Gumbel}&\multirow{1}{*}{30}&     0.113&     0.113&     0.113&     0.106&     0.105&     0.108&     0.062&     0.086&     0.016\\
				\cline{5-13}
				&&&\multirow{1}{*}{50}&     0.113&     0.107&     0.112&     0.098&     0.099&     0.099&     0.058&     0.079&     0.015\\
				\cline{3-13}
				&&\multirow{2}{*}{Student}&\multirow{1}{*}{30}&     0.115&     0.114&     0.108&     0.100&     0.096&     0.099&     0.073&     0.092&     0.017\\
				\cline{5-13}
				&&&\multirow{1}{*}{50}&     0.116&     0.106&     0.105&     0.105&     0.106&     0.103&     0.059&     0.074&     0.018\\
				\cline{3-13}
				&&\multirow{2}{*}{Frank}&\multirow{1}{*}{30}&     0.103&     0.118&     0.113&     0.098&     0.100&     0.100&     0.061&     0.087&     0.017\\
				\cline{5-13}
				&&&\multirow{1}{*}{50}&     0.117&     0.113&     0.109&     0.101&     0.100&     0.097&     0.057&     0.074&     0.019\\
				\hline\hline
				
				\multicolumn{1}{c}{}&\multicolumn{1}{c}{}&&&\(R_n'\)&\(S_n'\)&\(T_n'\)&\(R_n'\)&\(S_n'\)&\(T_n'\)&\(R_n'\)&\(S_n'\)&\(T_n'\)\\
				\hline
				\multirow{20}{*}{\rotatebox[origin=c]{90}{Radial symmetry tests}}&\multirow{10}{*}{\rotatebox[origin=c]{90}{$\alpha=0.05$}}&\multirow{2}{*}{Gauss}&\multirow{1}{*}{30} &     0.051&     0.056&     0.059&     0.047&     0.046&     0.044&     0.005&     0.012&     0.006\\
				\cline{5-13}
				&&&\multirow{1}{*}{50}&     0.050&     0.052&     0.056&     0.047&     0.049&     0.049&     0.018&     0.025&     0.029\\
				\cline{3-13}
				&&\multirow{2}{*}{Student}&\multirow{1}{*}{30}&     0.057&     0.055&     0.054&     0.055&     0.055&     0.054&     0.009&     0.020&     0.005\\
				\cline{5-13}
				&&&\multirow{1}{*}{50}&     0.059&     0.059&     0.059&     0.044&     0.044&     0.049&     0.028&     0.032&     0.027\\
				\cline{3-13}
				&&\multirow{2}{*}{Frank}&\multirow{1}{*}{30}&     0.054&     0.055&     0.062&     0.050&     0.049&     0.051&     0.007&     0.008&     0.004\\
				\cline{5-13}
				&&&\multirow{1}{*}{50}&     0.052&     0.051&     0.054&     0.054&     0.053&     0.054&     0.024&     0.037&     0.037\\
				\cline{3-13}
				&&\multirow{2}{*}{Plackett}&\multirow{1}{*}{30}&     0.059&     0.058&     0.060&     0.053&     0.052&     0.053&     0.007&     0.018&     0.007\\
				\cline{5-13}
				&&&\multirow{1}{*}{50}&     0.048&     0.048&     0.054&     0.052&     0.052&     0.051&     0.024&     0.031&     0.031\\
				\cline{3-13}
				&&\multirow{2}{*}{Cauchy}&\multirow{1}{*}{30}&     0.077&     0.068&     0.061&     0.052&     0.055&     0.051&     0.019&     0.026&     0.002\\
				\cline{5-13}
				&&&\multirow{1}{*}{50}&     0.068&     0.066&     0.057&     0.054&     0.055&     0.052&     0.036&     0.040&     0.033\\
				\cline{2-13}
				&\multirow{10}{*}{\rotatebox[origin=c]{90}{$\alpha=0.1$}}&\multirow{2}{*}{Gauss}&\multirow{1}{*}{30} &     0.106&     0.112&     0.112&     0.098&     0.100&     0.098&     0.017&     0.033&     0.021\\
				\cline{5-13}
				&&&\multirow{1}{*}{50}&     0.101&     0.105&     0.107&     0.098&     0.099&     0.099&     0.047&     0.067&     0.076\\
				\cline{3-13}
				&&\multirow{2}{*}{Student}&\multirow{1}{*}{30}&     0.106&     0.108&     0.107&     0.107&     0.103&     0.105&     0.022&     0.056&     0.024\\
				\cline{5-13}
				&&&\multirow{1}{*}{50}&     0.113&     0.113&     0.112&     0.098&     0.099&     0.096&     0.059&     0.086&     0.066\\
				\cline{3-13}
				&&\multirow{2}{*}{Frank}&\multirow{1}{*}{30}&     0.110&     0.110&     0.106&     0.100&     0.101&     0.103&     0.008&     0.037&     0.016\\
				\cline{5-13}
				&&&\multirow{1}{*}{50}&     0.104&     0.104&     0.108&     0.108&     0.107&     0.103&     0.062&     0.076&     0.086\\
				\cline{3-13}
				&&\multirow{2}{*}{Plackett}&\multirow{1}{*}{30}&     0.110&     0.108&     0.107&     0.101&     0.102&     0.099&     0.013&     0.039&     0.022\\
				\cline{5-13}
				&&&\multirow{1}{*}{50}&     0.098&     0.098&     0.099&     0.099&     0.095&     0.101&     0.056&     0.073&     0.079\\
				\cline{3-13}
				&&\multirow{2}{*}{Cauchy}&\multirow{1}{*}{30}&     0.134&     0.125&     0.111&     0.103&     0.104&     0.098&     0.030&     0.069&     0.020\\
				\cline{5-13}
				&&&\multirow{1}{*}{50}&     0.124&     0.120&     0.108&     0.110&     0.111&     0.108&     0.093&     0.097&     0.053\\
				\hline\hline
			\end{tabular}
		\end{center}
		\caption{Null rejection rates: exchangeability and radial symmetry tests using feasible randomization, infeasible randomization, and the multiplier bootstrap.}\label{nulltab}
	}

\end{table}

The top half of Table \ref{nulltab} displays results for testing the null of exchangeability using the statistics \(R_n\), \(S_n\) and \(T_n\). Results are reported for five exchangeable copulas: the Gaussian, Clayton, Gumbel, Student and Frank copulas, parameterized to have a rank correlation of 0.5 (and with 3 degrees of freedom for the Student copula). The first three columns of rejection rates correspond to our feasible randomization procedure. These rates are slightly greater than, but very close to, the nominal level, especially at the larger sample size \(n=50\). There is no appreciable difference between the rejection rates obtained using the three statistics \(R_n\), \(S_n\) and \(T_n\). The next three columns of rejection rates correspond to infeasible randomization applied directly to the unobserved pairs \((U_i,V_i)\). This testing procedure is exact, and indeed we see that the computed rejection rates are extremely close to the nominal level, differing only due to Monte Carlo error. The final three columns of rejection rates correspond to the multiplier bootstrap. These rejection rates are much lower than the nominal level. The lowest rejection rates are obtained using the \(T_n\) statistic, and the largest with the \(S_n\) statistic. Simulations reported by \citet[Table 1]{GNQ12} indicate that excessive conservatism may continue to be an issue at sample sizes as large as \(n=250\).

The bottom half of Table \ref{nulltab} displays results for testing the null of radial symmetry using the statistics \(R'_n\), \(S'_n\) and \(T'_n\). Results are reported for five radially symmetric copulas: the Gaussian, Student, Frank, Plackett and Cauchy copulas, parameterized to have a rank correlation of 0.5 (and with 3 degrees of freedom for the Student copula). Qualitatively, the results are similar to those reported for the exchangeability tests. Our feasible randomization procedure tends to overreject a bit more with the Cauchy copula than with the other copulas considered, especially with the statistic \(R_n'\) at the smaller sample size \(n=30\). Again, the rejection rates based on the multiplier bootstrap tend to be much lower than the nominal level, and simulations reported by \citet[Table 1]{GN14} indicate that the issue may persist at sample sizes as large as \(n=500\).

In addition to the simulations reported in Table \ref{nulltab}, we ran further simulations to investigate the null rejection rates obtained when Procedure \ref{rtest} is naively applied to the rank pairs \((U_{ni},V_{ni})\). We computed a rejection rate of zero for all configurations listed in Table \ref{nulltab}. This may be unsurprising in view of Figure \ref{naiveapprox}, where we saw that the upper quantiles of the sampling distributions of \(S_n\) and \(S_n'\) when \(C\) is the product copula lie far to the left of those in the (poorly) approximating distributions obtained through a naive application of Procedure \ref{rtest}.

As discussed in Section \ref{secresample}, applied researchers may be unwilling to randomize between rejection and nonrejection when the test statistic and critical value are equal, and instead prefer to only reject when the test statistic strictly exceeds the critical value. To investigate the impact of such a decision rule, we recomputed the rejection rates for the feasible and infeasible randomization tests in Table \ref{nulltab}, recording rejections only when a test statistic strictly exceeded the critical value. We found that the rejection rates using the $R_n$, $S_n$, $R_n'$ and $S_n'$ statistics were essentially unchanged. However, the rejection rates using the $T_n$ and $T_n'$ statistics dropped to about 0.04 with nominal level $\alpha=0.05$, and to about 0.08 with nominal level $\alpha=0.1$.

\subsection{Power}

In this section we report numerical evidence on the power of symmetry tests based on feasible randomization and on the multiplier bootstrap. The main finding is that power is greater with feasible randomization. The improved power appears to be driven by our finding in Section \ref{secsize} that feasible randomization delivers a rejection rate close to the nominal level even at small sample sizes, whereas the multiplier bootstrap produces excessively conservative tests. We focus below on sample sizes \(n=50,100,250\). The power advantage of tests based on feasible randomization appears to be smaller at larger sample sizes. This is not unexpected in view of results of \citet{R89}, who showed the asymptotic equivalence of randomization and bootstrap procedures in fairly general settings. Note, however, that simulation results reported by \citet{GN14} indicate that the excessive conservatism induced by the multiplier bootstrap may persist at sample sizes as large as \(n=500\).

In Table \ref{altextab} we report the results of numerical simulations used to investigate the power of exchangeability tests based on our feasible randomization procedure or on the multiplier bootstrap at small (\(n=50\)) and medium (\(n=100\)) sample sizes, with nominal level \(\alpha=0.05\). As in the previous subsection, we computed rejection frequencies over 10000 experimental replications, and in each replication computed critical values using \(N=250\) random draws from the group \(\mathbf G_n\) or bootstrap repetitions.

\begin{sidewaystable}[p]
	{\footnotesize
		\begin{center}
			\begin{tabular}{lll|ccc|ccc|ccc|ccc}
				\hline\hline
				\multirow{3}{*}{Copula} &\multirow{3}{*}{$\delta$} & \multirow{3}{*}{$\tau$}&\multicolumn{6}{c|}{$n=50$}&\multicolumn{6}{|c}{$n=100$}\\
				&&&\multicolumn{3}{c}{Feasible randomization}&\multicolumn{3}{c|}{Multiplier bootstrap}&\multicolumn{3}{|c}{Feasible randomization}&\multicolumn{3}{c}{Multiplier bootstrap}\\
				\cline{4-15}
				&&&$R_{n}$&$S_{n}$&$T_{n}$  &  $R_{n}$&$S_{n}$&$T_{n}$&$R_{n}$&$S_{n}$&$T_{n}$  &  $R_{n}$&$S_{n}$&$T_{n}$\\
				\hline\hline
				\multirow{9}{*}{K-Gauss}& \multirow{3}{*}{\(\delta=0.25\)}&0.5	&0.093	&0.082	&0.084	&0.027	&0.043	&0.008&0.146	&0.123	&0.085	&0.073&0.070	&0.070\\
				\multirow{9}{*}{}&\multirow{3}{*}{}&0.7	&0.267	&0.236	&0.140	&0.076	&0.141	&0.007&0.616	&0.555	&0.302	&0.359&0.414	&0.229\\
				\multirow{9}{*}{}&\multirow{3}{*}{}&0.9	&0.929	&0.874	&0.515	&0.420	&0.793	&0.006&0.997	&0.998	&0.917	&0.921 &0.992	&0.665\\
				\cline{4-15}
				\multirow{9}{*}{}&\multirow{3}{*}{\(\delta=0.5\)}&0.5	&0.122	&0.105	&0.089	&0.068	&0.077	&0.007&0.266	&0.231	&0.145	&0.138 &0.154	&0.095\\
				\multirow{9}{*}{}&\multirow{3}{*}{}&0.7	&0.482	&0.434	&0.263	&0.252	&0.283	&0.020&0.836	&0.803	&0.520	&0.729	&0.744	&0.435\\
				\multirow{9}{*}{}&\multirow{3}{*}{}&0.9	&0.941	&0.934	&0.649	&0.753	&0.893	&0.037&1.000	&1.000	&0.972	&0.998	&1.000	&0.855\\
				\cline{4-15}
				\multirow{9}{*}{}& \multirow{3}{*}{\(\delta=0.75\)}&0.5	&0.109	&0.095	&0.082	&0.074	&0.069	&0.006&0.208	&0.205	&0.153	&0.118	&0.131	&0.123\\
				\multirow{9}{*}{}&\multirow{3}{*}{}&0.7	&0.282	&0.264	&0.173	&0.165	&0.182	&0.011&0.572	&0.548	&0.327	&0.492	&0.498	&0.259\\
				\multirow{9}{*}{}&\multirow{3}{*}{}&0.9	&0.513	&0.524	&0.273	&0.391	&0.452	&0.024&0.886	&0.900 &0.644	&0.809	&0.833	&0.478\\
				\hline
				\multirow{9}{*}{K-Clayton}& \multirow{3}{*}{\(\delta=0.25\)} &0.5	&0.078	&0.084	&0.069	&0.025	&0.046	&0.002&0.137	&0.133	&0.098	&0.066	&0.077	&0.061\\
				\multirow{9}{*}{}&\multirow{3}{*}{}&0.7	&0.279	&0.274	&0.118	&0.101	&0.176	&0.004&0.641	&0.643	&0.329	&0.407	&0.543	&0.237\\
				\multirow{9}{*}{}&\multirow{3}{*}{}&0.9	&0.915	&0.872	&0.519	&0.422	&0.856	&0.010&0.995	&1.000	&0.943	&0.928	&0.999	&0.700\\
				\cline{4-15}
				\multirow{9}{*}{}&\multirow{3}{*}{\(\delta=0.5\)} &0.5	&0.108	&0.096	&0.061	&0.060	&0.055	&0.005&0.183	&0.174	&0.113	&0.114	&0.109	&0.085\\
				\multirow{9}{*}{}&\multirow{3}{*}{}&0.7	&0.386	&0.371	&0.183	&0.208	&0.240	&0.012&0.716	&0.703	&0.445	&0.573	&0.610	&0.380\\
				\multirow{9}{*}{}&\multirow{3}{*}{}&0.9	&0.917	&0.905	&0.628	&0.706	&0.852	&0.042&0.998	&0.997	&0.948	&0.998	&1.000	&0.838\\
				\cline{4-15}
				\multirow{9}{*}{}& \multirow{3}{*}{\(\delta=0.75\)} &0.5	&0.080	&0.071	&0.068	&0.037	&0.039	&0.003&0.089	&0.090	&0.096	&0.064	&0.057&0.064\\
				\multirow{9}{*}{}&\multirow{3}{*}{}&0.7	&0.188	&0.172	&0.121	&0.084	&0.089	&0.013&0.351	&0.357	&0.225	&0.234	&0.241&0.141\\
				\multirow{9}{*}{}&\multirow{3}{*}{}&0.9	&0.441	&0.434	&0.243	&0.308	&0.333	&0.033&0.803	&0.795	&0.551	&0.695	&0.730	&0.429\\
				\hline
				\multirow{9}{*}{K-Gumbel}& \multirow{3}{*}{\(\delta=0.25\)} &0.5	&0.094	&0.078	&0.070	&0.034	&0.052	&0.004&0.173	&0.161	&0.115	&0.084	&0.084	&0.081\\
				\multirow{9}{*}{}&\multirow{3}{*}{}&0.7	&0.309	&0.241	&0.166	&0.108	&0.153	&0.003&0.599	&0.555	&0.336	&0.403	&0.460	&0.227\\
				\multirow{9}{*}{}&\multirow{3}{*}{}&0.9	&0.933	&0.882	&0.530	&0.409	&0.813	&0.005&0.995	&0.997	&0.925	&0.933	&0.996	&0.678\\
				\cline{4-15}
				\multirow{9}{*}{}&\multirow{3}{*}{\(\delta=0.5\)}&0.5	&0.184	&0.167	&0.123	&0.080	&0.086	&0.007&0.334	&0.325	&0.193	&0.238	&0.264	&0.164\\
				\multirow{9}{*}{}&\multirow{3}{*}{}&0.7	&0.536	&0.507	&0.323	&0.329	&0.353	&0.024&0.935	&0.925	&0.672	&0.784	&0.853	&0.470\\
				\multirow{9}{*}{}&\multirow{3}{*}{}&0.9	&0.961	&0.960	&0.702	&0.777	&0.922	&0.036&1.000	&1.000	&0.977	&0.999	&1.000	&0.869\\
				\cline{4-15}
				\multirow{9}{*}{}& \multirow{3}{*}{\(\delta=0.75\)}&0.5	&0.189	&0.179	&0.140	&0.089	&0.101	&0.011&0.319	&0.313	&0.192	&0.239	&0.231	&0.132\\
				\multirow{9}{*}{}&\multirow{3}{*}{}&0.7	&0.363	&0.352	&0.235	&0.253	&0.266	&0.021&0.707	&0.710	&0.421	&0.601	&0.597	&0.343\\
				\multirow{9}{*}{}&\multirow{3}{*}{}&0.9	&0.518	&0.550	&0.291	&0.381	&0.465	&0.026&0.934	&0.921	&0.622	&0.817	&0.861	&0.473\\
				\hline\hline
			\end{tabular}%
		\end{center}
		\caption{Power: exchangeability tests using feasible randomization and the multiplier bootstrap.}\label{altextab}
	}
\end{sidewaystable}	
	
	The nonexchangeable copulas used to produce the rejection frequencies in Table \ref{altextab} were obtained by applying the Khoudraji transform \citep{K95} to the Gaussian, Clayton and Gumbel copulas with parameters chosen such that the rank correlation \(\tau\) is equal to 0.5, 0.7 or 0.9. The Khoudraji transform of a copula \(C\) is given by \(C^\mathrm{K}_\delta(u,v)=u^\delta C(u^{1-\delta},v)\), where we allow \(\delta=0.25,0.5,0.75\). \citet{GNQ12} also simulate rejection frequencies of their exchangeability tests using these copula specifications; the rejection rates for the multiplier bootstrap with \(n=100\) in Table \ref{altextab} are taken directly from their paper.
	
	Cursory examination of the numbers in Table \ref{altextab} reveals that our feasible randomization procedure generates more power than the multiplier bootstrap -- often much more. A test using the statistic \(T_n\) with the multiplier bootstrap has essentially no power at the smaller sample size \(n=50\), whereas a test using the same statistic with our feasible randomization procedure has substantial power. This may be unsurprising in view of the extremely low null rejection rates exhibited by the former test in Table \ref{nulltab}. The differences in power with the statistics \(R_n\) and \(S_n\), or at the larger sample size \(n=100\), are less extreme, but nevertheless it is clear that our randomization procedure generates significant power improvements over the multiplier bootstrap at these sample sizes, with the latter procedure hamstrung by its low null rejection rates.
	
	\begin{table}[t]
		\footnotesize
		\begin{center}
			\begin{tabular}{lll|ccc|lll}
				\hline\hline
				\multirow{2}{*}{Copula} &\multirow{2}{*}{$n$} & \multirow{2}{*}{$\tau$}&\multicolumn{3}{c}{Feasible randomization}&\multicolumn{3}{c}{Multiplier bootstrap}\\
				\cline{4-9}
				{}&{}&{}&$R_{n}'$&$S_{n}'$&$T_{n}'$  &  $R_{n}'$&$S_{n}'$&$T_{n}'$\\
				\hline\hline
				\multirow{9}{*}{Clayton}& \multirow{3}{*}{50} &0.25	&0.236	&0.259	&0.172	& 0.102& 0.156&0.092\\
				\multirow{9}{*}{}&\multirow{3}{*}{}&0.5&0.428	&0.514	&0.325& 0.095& 0.174&0.106\\
				\multirow{9}{*}{}&\multirow{3}{*}{}&0.75&0.397	&0.508&0.354& 0.001& 0.029& 0.009\\
				\cline{4-9}
				\multirow{9}{*}{}&\multirow{3}{*}{100}&0.25&0.374	&0.412&0.269& 0.242& 0.351& 0.184\\
				\multirow{9}{*}{}&\multirow{3}{*}{}&0.5&0.745&0.808	&0.625& 0.470& 0.713 &0.459\\
				\multirow{9}{*}{}&\multirow{3}{*}{}&0.75&0.859	&0.924&0.762& 0.177& 0.728& 0.396\\
				\cline{4-9}
				\multirow{9}{*}{}&\multirow{3}{*}{250}&0.25&0.748&0.803&0.618&0.645&0.723&0.459\\
				\multirow{9}{*}{}&\multirow{3}{*}{}&0.5&0.998&1.000&0.973&0.984&0.995&0.953\\
				\multirow{9}{*}{}&\multirow{3}{*}{}&0.75&1.000&1.000&0.998&0.995&1.000&0.985\\
				\hline
				\multirow{9}{*}{Gumbel}& \multirow{3}{*}{50} &0.25	&0.054&0.059	&0.053	& 0.053& 0.056&0.056\\
				\multirow{9}{*}{}&\multirow{3}{*}{}&0.5&0.075	&0.064	&0.061 & 0.010& 0.017&0.027\\
				\multirow{9}{*}{}&\multirow{3}{*}{}&0.75&0.064	&0.066 &0.083& 0.002& 0.002& 0.001\\
				\cline{4-9}
				\multirow{9}{*}{}&\multirow{3}{*}{100}&0.25&0.093	&0.080&0.077& 0.078& 0.079& 0.059\\
				\multirow{9}{*}{}&\multirow{3}{*}{}&0.5&0.163&0.154	&0.128& 0.086& 0.122 &0.085\\
				\multirow{9}{*}{}&\multirow{3}{*}{}&0.75&0.145&0.148&0.107& 0.005& 0.058& 0.031\\
				\cline{4-9}
				\multirow{9}{*}{}&\multirow{3}{*}{250}&0.25&0.225&0.218&0.163&0.195&0.192 &0.132\\
				\multirow{9}{*}{}&\multirow{3}{*}{}&0.5&0.413&0.433&0.323&0.383&0.429&0.312\\
				\multirow{9}{*}{}&\multirow{3}{*}{}&0.75&0.403&0.488&0.278&0.235&0.364&0.215\\
				\hline\hline
			\end{tabular}%
		\end{center}
		\caption{Power: radial symmetry tests using feasible randomization and the multiplier bootstrap.}\label{altrstab}
	\end{table}
	
In Table \ref{altrstab} we report the results of numerical simulations used to investigate the power of radial symmetry tests based on our feasible randomization procedure or on the multiplier bootstrap at small (\(n=50\)), medium (\(n=100\)) and large (\(n=250\)) sample sizes, with nominal level \(\alpha=0.05\). We again used 10000 experimental replications, and \(N=250\) repetitions for the randomization and bootstrap procedures in each replication. Rejection rates were computed for the (radially asymmetric) Clayton and Gumbel copulas, parametrized to have rank correlation \(\tau=0.25,0.5,0.75\). As was the case in Table \ref{altextab}, it is clear from the numbers in Table \ref{altrstab} that our feasible randomization procedure generates large improvements in power over the multiplier bootstrap. This is particularly true for the Clayton copula at the smaller sample sizes \(n=50\) and \(n=100\). Rejection rates are lower for the Gumbel copula, which apparently exhibits a more mild form of radial asymmetry than the Clayton copula.

\subsection{Simulations with conditionally heteroskedastic data}

In Section \ref{secintro} we mentioned that radially asymmetric dependence between asset returns has been a subject of interest in empirical finance. As is well known, time series of asset returns typically exhibit conditional heteroskedasticity. A bivariate time series of asset returns will therefore not satisfy the iid condition underlying our analysis. We may nevertheless consider applying our symmetry tests to bivariate return series after first filtering out conditional heteroskedasticity or other dependencies in the data. In this section we report the results of some simulations we ran to investigate this possibility.

The design of our simulations is motivated by the semiparametric copula-based multivariate dynamic (SCOMDY) model of \citet{CF06}. In this model, a pair (or more) of time series are each assumed to evolve according to an ARMA-GARCH model or some other parametric conditional mean-variance specification. Dependence between the two series is generated by linking contemporaneous ARMA-GARCH innovations with a parametric copula function. The marginal distributions of the innovations are left unspecified, making the SCOMDY model semiparametric.

We calibrated our simulations using 658 daily returns on the Hang Seng and Nikkei market indices, running from March 7, 2016, to March 5, 2019. A scatterplot of these returns is displayed in panel (a) of Figure \ref{scatterplots}. To each univariate return series we fit a GARCH(1,1) model with iid Student innovations by maximum likelihood. A scatterplot of the fitted innovations for each return series is displayed in panel (b) of Figure \ref{scatterplots}. The fitted innovations were transformed to ranks, and divided by sample size to render them between zero and one. A scatterplot of the ranked fitted innovations is displayed in panel (c) of Figure \ref{scatterplots}. This last scatterplot provides us with a sense of the shape of the copula linking contemporaneous GARCH innovations. At a glance, it does not appear to provide evidence of nonexchangeability or radial asymmetry. To confirm, we applied our feasible randomization tests of copula exchangeability and radial symmetry to the fitted innovations, using \(N=1000\) repetitions of our feasible randomization procedure. We computed exchangeability test statistics \(R_n=0.0262\), \(S_n=0.0283\) and \(T_n=0.6627\), none of which led us to reject the null at the 0.2 nominal level. We computed radial symmetry test statistics \(R'_n=0.0459\), \(S'_n=0.0581\) and \(T'_n=0.7407\), none of which led us to reject the null at the 0.15 nominal level. We would therefore like to conclude that the GARCH innovations are consistent with copula exchangeability and radial symmetry. However, such a conclusion is not justified by the results of this paper, because the symmetry tests were applied to fitted GARCH innovations, and it has not been established whether the preliminary estimation of GARCH parameters affects the asymptotic null rejection rates of our testing procedures.

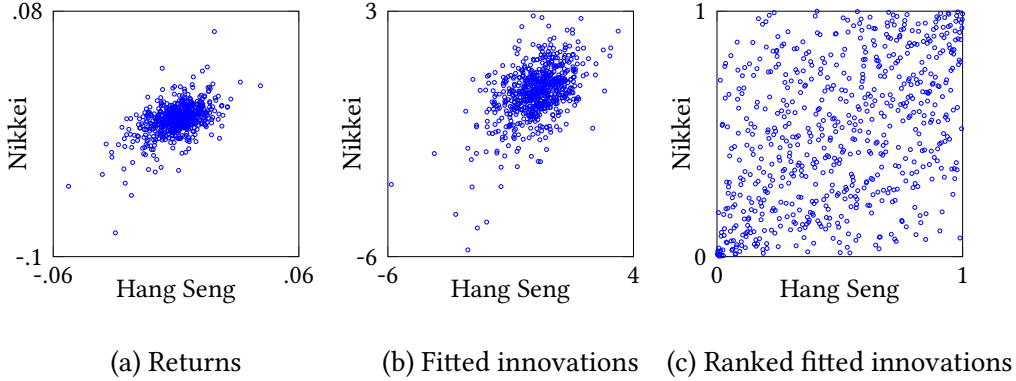
\begin{figure}[]
	\begin{center}
		\pgfplotsset{
			tick label style={font=\small},
			label style={font=\small},
			legend style={font=\footnotesize},
			scaled y ticks=false
		}
		\begin{tikzpicture}[scale=0.925]
		\begin{axis}[height=0.35\textwidth, width=0.35\textwidth, /pgf/number format/1000 sep={}, enlargelimits=false, ytick={-.1,.08}, yticklabels={-.1,.08}, ymin=-.1, ymax=.08, xtick={-.06,.06}, xticklabels={-.06,.06}, xmin=-.06, xmax=.06, xlabel={Hang Seng}, ylabel={Nikkei}, ticklabel style={/pgf/number format/fixed}, x label style={at={(axis description cs:0.5,0.1)}}, y label style={at={(axis description cs:0.2,0.5)}}, xtick scale label code/.code={}, title style={at={(0.5,-.4)},anchor=north,yshift=-.1}, title={(a) Returns}]
		\addplot+[mark=o,only marks,mark size=0.8pt]
		table [x=HSI, y=NI225] {returndata.csv};
		\end{axis}
		\end{tikzpicture}
		\begin{tikzpicture}[scale=0.925]
		\begin{axis}[height=0.35\textwidth, width=0.35\textwidth, /pgf/number format/1000 sep={}, enlargelimits=false, ytick={-6,3}, yticklabels={-6,3}, ymin=-6, ymax=3, xtick={-6,4}, xticklabels={-6,4}, xmin=-6, xmax=4, xlabel={Hang Seng}, ylabel={Nikkei}, x label style={at={(axis description cs:0.5,0.1)}}, y label style={at={(axis description cs:0.2,0.5)}}, title style={at={(0.5,-.4)},anchor=north,yshift=-.1}, title={(b) Fitted innovations}]
		\addplot+[mark=o,only marks,mark size=0.8pt]
		table [x=HSIres, y=NI225res] {returndata.csv};
		\end{axis}
		\end{tikzpicture}
		\begin{tikzpicture}[scale=0.925]
		\begin{axis}[height=0.35\textwidth, width=0.35\textwidth, /pgf/number format/1000 sep={}, enlargelimits=false, ytick={0,1}, ymin=0, ymax=1, xtick={0,1}, xmin=0, xmax=1, xlabel={Hang Seng}, ylabel={Nikkei}, x label style={at={(axis description cs:0.5,0.1)}}, y label style={at={(axis description cs:0.2,0.5)}}, title style={at={(0.5,-.4)},anchor=north,yshift=-.1}, title={(c) Ranked fitted innovations}]
		\addplot+[mark=o,only marks,mark size=0.8pt]
		table [x=HSIresrank, y=NI225resrank] {returndata.csv};
		\end{axis}
		\end{tikzpicture}
		\caption{Hang Seng and Nikkei market index returns, 3/7/2016--3/5/2019. Panel (a) displays raw daily returns. Panel (b) displays GARCH(1,1) fitted innovations. Panel (c) displays GARCH(1,1) fitted innovations after transforming to ranks and dividing by sample size.}
		\label{scatterplots}
	\end{center}
\end{figure}

To investigate further, we ran a number of simulations based on the SCOMDY model fitted to Hang Seng and Nikkei returns. Rather than using a parametric copula specification, we took the following nonparametric approach: we first obtained a bivariate distribution \(H\) for the GARCH innovations by slightly smoothing the empirical distribution of the fitted innovations in panel (b) of Figure \ref{scatterplots}. Let \(F\) and \(G\) be the margins and \(C\) the copula corresponding to \(H\). To impose the null of copula exchangeability or radial symmetry in our simulations, we set \(H^\ast=C^\ast(F,G)\), where \(C^\ast=(C+C^\top)/2\) when we wish to impose exchangeability, and \(C^\ast=(C+C^\rot)/2\) when we wish to impose radial symmetry. The data generating process for our simulations is then as follows: in each of \(1000\) experimental replications, we generate \(n\) pairs of innovations by iid sampling from \(H^\ast\), and then feed these into the GARCH(1,1) models calibrated to the Hang Seng and Nikkei returns, to produce simulated returns \((X_1,Y_1),\ldots,(X_n,Y_n)\). We then re-estimate the GARCH(1,1) models using these simulated returns, and apply our symmetry tests to the pairs of fitted innovations with \(N=250\) repetitions of our randomization procedure. To provide a point of comparison, we also applied our symmetry tests to the true innovations drawn directly from \(H^\ast\). We used sample sizes \(n\) of 50, 100 and 250, and nominal significance levels \(\alpha\) of 0.05 and 0.1.

\begin{table}[]
	\vspace*{-3\baselineskip}
	{\footnotesize
		\begin{center}
			\begin{tabular}{c|c|c|ccc|ccc}
				\hline\hline
				\multicolumn{1}{c}{}&\multicolumn{1}{c}{}& \multirow{2}{*}{}&\multicolumn{3}{c}{\thead{Fitted innovations}} &\multicolumn{3}{c}{\thead{True innovations}}\\
				\cline{4-9}
				\multicolumn{1}{c}{}&$\alpha$&$n$& $R_{n}$ & $S_{n}$ & $T_{n}$ & $R_{n}$ & $S_{n}$ & $T_{n}$ \\
				\hline\hline
				\multirow{6}{*}{\rotatebox[origin=c]{90}{\thead{Exchangeability\\ tests}}}&&50 &  0.061   &  0.066   &  0.058   &  0.061   &  0.063   &  0.054   \\
				&0.05&100&  0.058   &  0.062   &  0.064   &  0.060   &  0.064   &  0.070   \\
				&&250&  0.059   &  0.061   &  0.054   &  0.057   &  0.068   &  0.053   \\
				\cline{2-9}
				&&50 &  0.115   &  0.122   &  0.131   &  0.114   &  0.106   &  0.123   \\
				&0.1&100&  0.114   &  0.119   &  0.103  &  0.116   & 0.113   &  0.117   \\
				&&250&  0.117   &  0.106   &  0.095   &  0.123   &  0.113   &  0.099   \\
				\hline\hline
				\multicolumn{1}{c}{}&\multicolumn{1}{c}{}&&\(R_n'\)&\(S_n'\)&\(T_n'\)&\(R_n'\)&\(S_n'\)&\(T_n'\)\\
				\hline
				\multirow{6}{*}{\rotatebox[origin=c]{90}{\thead{Radial\\ symmetry tests}}}&&50 &  0.056   &  0.054   &  0.068   &  0.055   &  0.050   &  0.060   \\
				&0.05&100&  0.041  &  0.037   &  0.047   &  0.046   &  0.041   &  0.044   \\
				&&250&  0.039   &  0.047   &   0.056  &  0.050   &  0.052   &  0.049   \\
				\cline{2-9}
				&&50 &  0.115   &  0.118   &  0.109   &  0.120   &  0.117   &  0.102  \\
				&0.1&100&  0.082   &  0.086  &  0.098   &  0.090   &  0.094  &  0.098   \\
				&&250&  0.104   &  0.101   &  0.107   &  0.099   &  0.106   &  0.106   \\
				\hline\hline
			\end{tabular}
		\end{center}
		\caption{Null rejection rates: exchangeability and radial symmetry tests applied to the fitted and unobserved innovations of a pair of GARCH(1,1) processes calibrated to daily Hang Seng and Nikkei market index returns.}\label{GARCHtab}
	}
	
\end{table}

The results of our simulations are displayed in Table \ref{GARCHtab}. Strikingly, we see that the null rejection rates when our symmetry tests are applied to the fitted innovations are not appreciably different to the rejection rates obtained with the true innovations, and are generally close to the nominal level. Thus, the preliminary estimation of GARCH parameters does not appear to be affecting the rejection rate of our tests. A similar phenomenon has been shown to hold in closely related contexts. \citet{CF06} showed that, when the copula for contemporaneous innovations belongs to a known parametric class, the asymptotic distribution of the maximum likelihood estimator of the copula parameters is unaffected by the preliminary estimation of GARCH parameters. \citet{CCCFP09} showed that the asymptotic distribution of a statistic measuring copula goodness-of-fit is also unaffected by the preliminary estimation of GARCH parameters. Our tests of symmetry appear to benefit from the same property, although we have not shown this formally.

\section{Final remarks}\label{conclusion}

We have proposed a feasible randomization procedure which leads to consistent tests of copula symmetry which control size asymptotically. Numerical simulations indicate that basing tests of copula symmetry on our feasible randomization procedure instead of the multiplier bootstrap leads to substantially improved small sample performance.

A number of avenues for further research remain open. The asymptotic validity of our procedure was demonstrated for the specific statistics \(R_n\), \(S_n\), \(T_n\), \(R'_n\), \(S'_n\) and \(T'_n\), but in principle we may apply our procedure using other statistics measuring copula asymmetry. For instance, \citet{BQ17} have recently proposed tests of copula radial symmetry based on a range of statistics computed from the copula characteristic function, with critical values obtained using an adaptation of the multiplier bootstrap. Like the radial symmetry tests of \citet{GN14}, the tests of Bahraoui and Quessy frequently have null rejection rates well below the nominal level at smaller sample sizes; they note that ``this behavior is typical of methods based on the multiplier bootstrap'' \citep[p.\ 2075]{BQ17}. Further investigation of the applicability of our feasible randomization procedure with these and other statistics measuring nonexchangeability or radial asymmetry remains a task for future research.

It may be useful to extend our feasible randomization procedure to other hypothesis testing problems in which the need to estimate margins is a complicating factor. \citet{S18} has recently taken this approach to develop a permutation test of the null hypothesis that two independent multivariate samples are drawn from multivariate distributions with the same copula. We might also consider adapting our procedure to obtain tests of forms of copula symmetry other than exchangeability or radial symmetry, such as joint symmetry \citep[p.\ 36]{N06}, which is a stronger property than radial symmetry. We leave the pursuit of such adaptations for future research.

\appendix

\section{Proofs and supplementary lemmas}\label{appx}

\begin{lemma}\label{frwlem}
	\(C_n\) differs from \(C_n^\mathrm{D}\) by no more than \(2n^{-1}\) a.s.
\end{lemma}

\begin{proof}
	See \citet[p.\ 854]{FRW04}.	
\end{proof}

\begin{lemma}\label{gnlem}
	\(D_n\) differs from \(C_n^\mathrm{s}\) by no more than \(4n^{-1}\) a.s.
\end{lemma}

\begin{proof}
	\citet[pp.\ 1110]{GN14}	point out that \(C_n(u,v)-D_n(u,v)=E_n(u,v)-\lambda_n(u,v)\), where
	\begin{equation*}
	E_n(u,v)=C_n(u,v)-C_n(1,1)+C_n(1-u,1)+C_n(1,1-v)-C_n(1-u,1-v)
	\end{equation*}
	and
	\begin{equation*}
	\lambda_n(u,v)=\frac{1}{n}\sum_{i=1}^n\mathbbm{1}(U_{ni}=1-u,V_{ni}\geq1-v)+\frac{1}{n}\sum_{i=1}^n\mathbbm{1}(U_{ni}>1-u,V_{ni}=1-v).
	\end{equation*}
	It follows that
	\begin{align*}
	D_n(u,v)-C_n^\srv(u,v)&=\lambda_n(u,v)-C_n^\srv(u,v)+C_n(1,1)-C_n(1-u,1)-C_n(1,1-v)\\&\quad+C_n(1-u,1-v),
	\end{align*}
	which simplifies to
	\begin{equation*}
	D_n(u,v)-C_n^\srv(u,v)=\lambda_n(u,v)+2-u-v-C_n(1-u,1)-C_n(1,1-v).
	\end{equation*}
	In the probability one event that there are no ties between \(X_i\)'s or between \(Y_i\)'s, we have \(|C_n(1-u,1)-1+u|\leq n^{-1}\), \(|C_n(1,1-v)-1+v|\leq n^{-1}\), and \(|\lambda_n(u,v)|\leq 2n^{-1}\).
\end{proof}

\begin{lemma}\label{lemeta}
	Let \(\tilde{C}_n^\tau\) be defined as in \eqref{Ftauex} with \(\pi^0\) and \(\pi^1\) the pair of maps defined in either \eqref{expi} or \eqref{rspi}, and let
	\begin{equation*}
	\check{C}_n^\tau(u,v)=n^{-1}\sum_{i=1}^n\mathbbm{1}(\check{U}_{ni}^\tau\leq u,\check{V}_{ni}^\tau\leq v),
	\end{equation*}
	where \((\check{U}_{ni}^\tau,\check{V}_{ni}^\tau)=\pi^{\tau_i}(U_{ni},V_{ni})-n^{-1}(\eta_i,\eta_i)\), and \(\eta_1,\ldots,\eta_n\) are independent random variables distributed uniformly on \((0,1)\). Then \(\check{C}_n^\tau\) differs from \(\tilde{C}_n^\tau\) on \([0,1]^2\) by no more than \(4n^{-1}\) a.s.
\end{lemma}

\begin{proof}
    In the probability one event that there are no ties between \(X_i\)'s or between \(Y_i\)'s, we know that \((U_{n1},\ldots,U_{nn})\) and \((V_{n1},\ldots,V_{nn})\) are each equal to some permutation of \((1/n,2/n,\ldots,1)\). Thus the \(U_{ni}^\tau\)'s and \(V_{ni}^\tau\)'s can only take the values \(1/n,\ldots,1\), and also 0 if \(\pi^0\) and \(\pi^1\) are defined as in \eqref{rspi}, with at most two of the \(U_{ni}^\tau\)'s and at most two of the \(V_{ni}^\tau\)'s sharing any given value. Since \(-1/n<-\eta_i/n<0\) for each \(i\), it follows that for any fixed \((u,v)\), we may have \(\mathbbm{1}(\check{U}_{ni}^\tau\leq u,\check{V}_{ni}^\tau\leq v)\) and \(\mathbbm{1}({U}_{ni}^\tau\leq u,{V}_{ni}^\tau\leq v)\) differing for at most four values of \(i\). Thus
    \begin{equation*}
    |\check{C}_n^\tau(u,v)-\tilde{C}_n^\tau(u,v)|\leq n^{-1}\sum_{i=1}^n\left|\mathbbm{1}(\check{U}_{ni}^\tau\leq u,\check{V}_{ni}^\tau\leq v)-\mathbbm{1}({U}_{ni}^\tau\leq u,{V}_{ni}^\tau\leq v)\right|\leq4n^{-1},
    \end{equation*}
    as claimed.
\end{proof}

\begin{lemma}\label{empboxc}
	For any \(u,u',v,v'\in[0,1]\) we have
	\begin{equation*}
	\frac{1}{n}\sum_{i=1}^n\mathbbm{1}(u\leq U_{ni}\leq u',v\leq V_{ni}\leq v')\to\pr(u<U\leq u',v<V\leq v')\text{ a.s.}
	\end{equation*}
\end{lemma}

\begin{proof}
	Observe that
	\begin{align*}
	\mathbbm{1}(u\leq U_{ni}\leq u',v\leq V_{ni}\leq v')&=\mathbbm{1}(U_{ni}\leq u',V_{ni}\leq v')-\mathbbm{1}(U_{ni}\leq u,V_{ni}\leq v')\\&\quad-\mathbbm{1}(U_{ni}\leq u',V_{ni}\leq v)+\mathbbm{1}(U_{ni}\leq u,V_{ni}\leq v)\\&\quad+\mathbbm{1}(U_{ni}=u,v\leq V_{ni}\leq v')+\mathbbm{1}(u<U_{ni}\leq u',V_{ni}=v).
	\end{align*}
	It follows that
	\begin{align*}
	\frac{1}{n}\sum_{i=1}^n\mathbbm{1}(u\leq U_{ni}\leq u',v\leq V_{ni}\leq v')&=C_n(u',v')-C_n(u,v')-C_n(u',v)+C_n(u,v)\\&\quad+\frac{1}{n}\sum_{i=1}^n\mathbbm{1}(U_{ni}=u,v\leq V_{ni}\leq v')\\&\quad+\frac{1}{n}\sum_{i=1}^n\mathbbm{1}(u<U_{ni}\leq u',V_{ni}=v).
	\end{align*}
	With probability one we cannot have \(U_{ni}=u\) for more than one \(i\), or \(V_{ni}=v\) for more than one \(i\). Therefore the last two terms on the right hand side of the last displayed equality are each bounded by \(n^{-1}\). Pointwise strong consistency of the empirical copula thus yields the a.s.\ convergence
	\begin{align*}
	\frac{1}{n}\sum_{i=1}^n\mathbbm{1}(u\leq U_{ni}\leq u',v\leq V_{ni}\leq v')&\to C(u',v')-C(u,v')-C(u',v)+C(u,v)\\&=\pr(u<U\leq u',v<V\leq v').
	\end{align*}
\end{proof}

\begin{proof}[Proof of Lemma \ref{lemcwc}]
To establish conditional weak convergence we first note that \(\xi^\tau_n\), as a map from the underlying probability space into \(\ell^\infty([0,1]^2)\), is simple and hence Borel measurable. The second condition in Definition \ref{cwcdef} is therefore automatically satisfied.
To show that the first condition is satisfied we will apply the functional central limit theorem of \citet[Thm.\ 10.6]{P90} as stated by \citet[Thm.\ 11.16]{K08}. Adopting the notation of those authors, we let \(f_{ni}(u,v)=n^{-1/2}\mathbbm{1}(\pi^{\tau_i}(U_{ni},V_{ni})\leq(u,v))\), so that \(\xi^\tau_n=\sum_{i=1}^n(f_{ni}-\ex_\tau f_{ni})\). We will verify that conditions (A)--(E) of Kosorok hold a.s.
\begin{asparaenum}
	\item[(A)] Manageability of the \(f_{ni}\)'s with envelopes \(F_{ni}=n^{-1/2}\) follows from the fact that \(f_{ni}(u,v)\) is always monotone in \((u,v)\), as discussed by \citet[p.\ 221]{K08}.
	\item[(B)] Using the fact that \(\ex_\tau Z_{ni}(u,v)Z_{nj}(u',v')=0\) for \(i\neq j\), we obtain
	\begin{align*}
	\ex_\tau\xi^\tau_n(u,v)\xi^\tau_n(u',v')&=\frac{1}{n}\sum_{i=1}^n\ex_\tau Z_{ni}(u,v)Z_{ni}(u',v').
	\end{align*}
	Further, by noting that
	\begin{equation*}
	Z_{ni}(u,v)=\frac{(-1)^{\tau_i}}{2}\left(\mathbbm{1}(\pi^0(U_{ni},V_{ni})\leq(u,v))-\mathbbm{1}(\pi^1(U_{ni},V_{ni})\leq(u,v))\right),
	\end{equation*}
	we obtain
	\begin{align*}
	Z_{ni}(u,v)Z_{ni}(u',v')&=\frac{1}{4}\mathbbm{1}(\pi^0(U_{ni},V_{ni})\leq(u\wedge u',v\wedge v'))\\&\quad+\frac{1}{4}\mathbbm{1}(\pi^1(U_{ni},V_{ni})\leq(u\wedge u',v\wedge v'))\\&\quad-\frac{1}{4}\mathbbm{1}(\pi^0(U_{ni},V_{ni})\leq(u,v),\,\pi^1(U_{ni},V_{ni})\leq(u',v'))\\&\quad-\frac{1}{4}\mathbbm{1}(\pi^1(U_{ni},V_{ni})\leq(u,v),\,\pi^0(U_{ni},V_{ni})\leq(u',v')),
	\end{align*}
	which does not depend on \(\tau_i\) and therefore equals \(\ex_\tau Z_{ni}(u,v)Z_{ni}(u',v')\). Suppose that \(\pi^0\) and \(\pi^1\) are defined as in \eqref{expi}. Then the equalities just demonstrated imply that
	\begin{align*}
	\ex_\tau\xi^\tau_n(u,v)\xi^\tau_n(u',v')&=\frac{1}{4}C_n(u\wedge u',v\wedge v')+\frac{1}{4}C_n(v\wedge v',u\wedge u')\\&\quad-\frac{1}{4}C_n(u\wedge v',v\wedge u')-\frac{1}{4}C_n(v\wedge u',u\wedge v').
	\end{align*}
	The pointwise strong consistency of the empirical copula therefore implies that condition (B) of Kosorok is satisfied a.s.\ with
	\begin{align}
	\lim_{n\to\infty}\ex_\tau\xi^\tau_n(u,v)\xi^\tau_n(u',v')&=\frac{1}{4}C(u\wedge u',v\wedge v')+\frac{1}{4}C(v\wedge v',u\wedge u')\notag\\&\quad-\frac{1}{4}C(u\wedge v',v\wedge u')-\frac{1}{4}C(v\wedge u',u\wedge v').\label{kerAexpf}
	\end{align}
	Suppose instead that \(\pi^0\) and \(\pi^1\) are defined as in \eqref{rspi}. In this case we have
	\begin{align*}
	\ex_\tau\xi^\tau_n(u,v)\xi^\tau_n(u',v')&=\frac{1}{4}C_n(u\wedge u',v\wedge v')+\frac{1}{4}D_n(u\wedge u',v\wedge v')\\&\quad-\frac{1}{4n}\sum_{i=1}^n\mathbbm{1}(1-u'\leq U_{ni}\leq u,1-v'\leq V_{ni}\leq v)\\
	&\quad-\frac{1}{4n}\sum_{i=1}^n\mathbbm{1}(1-u\leq U_{ni}\leq u',1-v\leq V_{ni}\leq v'),
	\end{align*}
	It now follows from the pointwise strong consistency of the empirical copula and Lemmas \ref{gnlem} and \ref{empboxc} that condition (B) of Kosorok is satisfied a.s.\ with
	\begin{align}
	\lim_{n\to\infty}\ex_\tau\xi^\tau_n(u,v)\xi^\tau_n(u',v')&=\frac{1}{4}C(u\wedge u',v\wedge v')+\frac{1}{4}C^\srv(u\wedge u',v\wedge v')\notag\\&\quad-\frac{1}{4}\mathrm{P}\left(1-u'<U\leq u,1-v'<V\leq v\right)\notag\\&\quad-\frac{1}{4}\mathrm{P}\left(1-u<U\leq u',1-v<V\leq v'\right).\label{kerArspf}
	\end{align}
	\item[(C)] \(\limsup_{n\to\infty}\sum_{i=1}^n\ex_\tau F_{ni}^2=1<\infty\), trivially.
	\item[(D)] \(\sum_{i=1}^n\ex_\tau F_{ni}^2\mathbbm{1}(F_{ni}>\epsilon)=\mathbbm{1}(n^{-1/2}>\epsilon)\to0\) for each \(\epsilon>0\), also trivially.
	\item[(E)] The quantity \(n|f_{ni}(u,v)-f_{ni}(u',v')|^2\) is equal to one if exactly one of the inequalities \(\pi^{\tau_i}(U_{ni},V_{ni})\leq(u,v)\) and \(\pi^{\tau_i}(U_{ni},V_{ni})\leq(u',v')\) is satisfied, and is equal to zero otherwise. We therefore have
	\begin{align}
	n|f_{ni}(u,v)-f_{ni}(u',v')|^2&=\mathbbm{1}(\pi^{\tau_i}(U_{ni},V_{ni})\leq(u,v))(1-\mathbbm{1}(\pi^{\tau_i}(U_{ni},V_{ni})\leq(u',v')))\notag\\&\quad+(1-\mathbbm{1}(\pi^{\tau_i}(U_{ni},V_{ni})\leq(u,v)))\mathbbm{1}(\pi^{\tau_i}(U_{ni},V_{ni})\leq(u',v'))\notag\\
	&=\mathbbm{1}\left(\pi^{\tau_i}(U_{ni},V_{ni})\leq(u,v)\right)+\mathbbm{1}\left(\pi^{\tau_i}(U_{ni},V_{ni})\leq(u',v')\right)\notag\\
	&\quad-2\cdot\mathbbm{1}\left(\pi^{\tau_i}(U_{ni},V_{ni})\leq(u\wedge u',v\wedge v')\right).\label{kosE0}
	\end{align}
	It follows that
	\begin{align}
	n\ex_\tau|f_{ni}(u,v)-f_{ni}(u',v')|^2&=\frac{1}{2}\mathbbm{1}\left(\pi^0(U_{ni},V_{ni})\leq(u,v)\right)+\frac{1}{2}\mathbbm{1}\left(\pi^0(U_{ni},V_{ni})\leq(u',v')\right)\notag\\
	&\quad+\frac{1}{2}\mathbbm{1}\left(\pi^1(U_{ni},V_{ni})\leq(u,v)\right)+\frac{1}{2}\mathbbm{1}\left(\pi^1(U_{ni},V_{ni})\leq(u',v')\right)\notag\\
	&\quad-\mathbbm{1}\left(\pi^0(U_{ni},V_{ni})\leq(u\wedge u',v\wedge v')\right)\notag\\&\quad-\mathbbm{1}\left(\pi^1(U_{ni},V_{ni})\leq(u\wedge u',v\wedge v')\right).\label{kosE}
	\end{align}
	Define
	\begin{equation*}
	\rho_n((u,v),(u',v'))=\left(\sum_{i=1}^n\ex_\tau|f_{ni}(u,v)-f_{ni}(u',v')|^2\right)^{1/2}.
	\end{equation*}
	Suppose that \(\pi^0\) and \(\pi^1\) are defined as in \eqref{expi}. Then from \eqref{kosE} we have
	\begin{align*}
	\rho_n((u,v),(u',v'))^2&=\frac{1}{2}C_n(u,v)+\frac{1}{2}C_n(u',v')+\frac{1}{2}C_n(v,u)+\frac{1}{2}C_n(v',u')\\&\quad-C_n(u\wedge u',v\wedge v')-C_n(v\wedge v',u\wedge u'),
	\end{align*}
	and so the uniform strong consistency of the empirical copula ensures that condition (E) of Kosorok is satisfied a.s.\ with
	\begin{align*}
	\rho((u,v),(u',v'))^2&=\frac{1}{2}C(u,v)+\frac{1}{2}C(u',v')+\frac{1}{2}C(v,u)+\frac{1}{2}C(v',u')\\&\quad-C(u\wedge u',v\wedge v')-C(v\wedge v',u\wedge u').
	\end{align*}
	Suppose instead that \(\pi^0\) and \(\pi^1\) are defined as in \eqref{rspi}. Then from \eqref{kosE} we have
	\begin{align}
	\rho_n((u,v),(u',v'))^2&=\frac{1}{2}C_n(u,v)+\frac{1}{2}C_n(u',v')+\frac{1}{2}D_n(u,v)+\frac{1}{2}D_n(u',v')\notag\\&\quad-C_n(u\wedge u',v\wedge v')-D_n(u\wedge u',v\wedge v'),
	\end{align}
	and so the uniform strong consistency of the empirical copula along with Lemma \ref{gnlem} ensures that condition (E) of Kosorok is satisfied a.s.\ with
	\begin{align*}
	\rho((u,v),(u',v'))^2&=\frac{1}{2}C(u,v)+\frac{1}{2}C(u',v')+\frac{1}{2}C^\srv(u,v)+\frac{1}{2}C^\srv(u',v')\notag\\&\quad-C(u\wedge u',v\wedge v')-C^\srv(u\wedge u',v\wedge v').
	\end{align*}
\end{asparaenum}

Having verified that Kosorok's conditions (A)--(E) hold a.s., it remains to verify Kosorok's almost measurable Suslin (AMS) condition. From \eqref{kosE0} we have
\begin{align*}
\sum_{i=1}^n|f_{ni}(u,v)-f_{ni}(u',v')|^2&=\frac{1}{n}\sum_{i=1}^n\mathbbm{1}\left(\pi^{\tau_i}(U_{ni},V_{ni})\leq(u,v)\right)\\&\quad+\frac{1}{n}\sum_{i=1}^n\mathbbm{1}\left(\pi^{\tau_i}(U_{ni},V_{ni})\leq(u',v')\right)\\&\quad-\frac{2}{n}\sum_{i=1}^n\mathbbm{1}\left(\pi^{\tau_i}(U_{ni},V_{ni})\leq(u\wedge u',v\wedge v')\right).
\end{align*}
With \(\pi^0\) and \(\pi^1\) defined as in either \eqref{expi} or \eqref{rspi}, \(\pi^{\tau_i}(U_{ni},V_{ni})\) takes values on the grid \(T_n=\{(i/n,j/n):0\leq i,j\leq n\}\). We therefore have \(\inf_{(u,v)\in T_n}\sum_{i=1}^n|f_{ni}(u,v)-f_{ni}(u',v')|^2=0\) for every \((u',v')\in[0,1]^2\), with the infimum achieved by choosing \((u,v)\in T_n\) to be the largest grid point in the rectangle \([0,u']\times[0,v']\). This proves separability of the \(f_{ni}\)'s, which by Lemma 11.15 of \citet{K08} is sufficient for the AMS condition. We have now verified that all conditions of Theorem 11.16 of \citet{K08} hold a.s. Its conclusion tells us that \(\sup_{f\in\mathrm{BL}_1(\ell^\infty([0,1]^2))}|\ex_\tau f(\xi^\tau_n)-\ex f(\mathbb A)|\to0\) a.s., where \(\mathbb A\) is a centered Gaussian random element with covariance kernel given by the expression on the right-hand side of equality \eqref{kerAexpf} when \(\pi^0\) and \(\pi^1\) are given by \eqref{expi}, and by the expression on the right-hand side of equality \eqref{kerArspf} when \(\pi^0\) and \(\pi^1\) are given by \eqref{rspi}. Moreover, the sample paths of \(\mathbb A\) are continuous with respect to the semimetric \(\rho\), hence also with respect to the Euclidean metric.
\end{proof}

\begin{proof}[Proof of Lemma \ref{exlem}]
	Define \(\tilde{C}_n^\tau\) as in \eqref{Ftauex} and observe that \(\sqrt{n}(\tilde{C}_n^\tau-\frac{1}{2}(C_n+C_n^\top))=n^{-1/2}\sum_{i=1}^nZ_{ni}\), with the summands \(Z_{ni}\) defined as in Lemma \ref{lemcwc}. From Lemmas \ref{lemcwc} and \ref{lemeta} we have
	\begin{equation}\label{cwcexpf}
	\sqrt{n}\left(\check{C}_n^\tau-\frac{1}{2}(C_n+C_n^\top)\right)\cwconv\mathbb A,
	\end{equation}
	where \(\mathbb A\) is a centered Gaussian process with covariance kernel given by \eqref{kerAex}. From the weak convergence \(\sqrt{n}(C_n-C)\rightsquigarrow\mathbb C\) and the continuous mapping theorem we also have
	\begin{equation}\label{uwcexpf}
	\sqrt{n}\left(\frac{1}{2}(C_n+C_n^\top)-\frac{1}{2}(C+C^\top)\right)\rightsquigarrow\frac{1}{2}(\mathbb C+\mathbb C^\top).
	\end{equation}
	
	Let \(\mathbf D_\Phi\) be the set of bivariate cdfs on \([0,1]^2\) with margins grounded at zero, and let \(\Phi:\mathbf D_\Phi\to\ell^\infty([0,1]^2)\) be the map that sends a cdf \(\tilde{H}\in\mathbf D_\Phi\) with margins \(\tilde{H}_1\) and \(\tilde{H}_2\) to \(\tilde{H}(\tilde{H}_1^\leftarrow,\tilde{H}_2^\leftarrow)\). Theorem 2.4 of \citet{BV13} establishes that $\Phi$ is Hadamard differentiable at any regular copula in \(\mathbf D_\Phi\) tangentially to
	\begin{equation*}
	\mathbf D_0=\left\{h\in\mathbf{C}:h(0,u)=h(u,0)=h(1,1)=0\text{ for all }u\in[0,1]\right\},
	\end{equation*}
	where \(\mathbf{C}\) is the space of continuous real valued functions on \([0,1]^2\). The copula \(\frac{1}{2}(C+C^\top)\in\mathbf D_\Phi\) inherits the property of regularity from the copula \(C\in\mathbf D_\Phi\). From B\"{u}cher and Volgushev's result we obtain the derivative of \(\Phi\) at \(\frac{1}{2}(C+C^\top)\) in direction \(h\in\mathbf D_0\):
	\begin{align*}
	\Phi'_{\frac{1}{2}(C+C^\top)}h(u,v)&=h(u,v)-\frac{1}{2}\left(\dot{C}_1(u,v)+\dot{C}_2(v,u)\right)h(u,1)\notag\\&\quad-\frac{1}{2}\left(\dot{C}_2(u,v)+\dot{C}_1(v,u)\right)h(1,v).
	\end{align*}
	Note that \(\mathbb A\) concentrates on \(\mathbf D_0\), and that \(\mathbb D=\Phi'_{\frac{1}{2}(C+C^\top)}\mathbb A\). In view of \eqref{cwcexpf} and \eqref{uwcexpf} we may therefore apply the conditional delta method to obtain
	\begin{equation*}
	\sqrt{n}\left(\Phi(\check{C}_n^\tau)-\Phi\left(\frac{1}{2}(C_n+C_n^\top)\right)\right)\cwconv\mathbb D.
	\end{equation*}
	Now, \(\Phi(\check{C}^\tau_n)\) is the Deheuvels empirical copula of the perturbed transformed rank pairs \((\check{U}^\tau_{ni},\check{V}^\tau_{ni})\), \(i=1,\ldots,n\), which differs from \(C_n^\tau\) by no more than \(2n^{-1}\) a.s.\ by Lemma \ref{frwlem}. Furthermore, using the fact that \(\frac{1}{2}(C_n+C_n^\top)\) has margins uniform on \(\{n^{-1},2n^{-1},\ldots,1\}\) a.s., it is easy to show that \(\Phi(\frac{1}{2}(C_n+C_n^\top))\) differs from \(\frac{1}{2}(C_n+C_n^\top)\) by no more than \(2n^{-1}\) a.s. We therefore have \eqref{excwc2} as claimed.
\end{proof}

\begin{lemma}\label{exlem2}
	If \(C=C^\top\) then the random element \(\mathbb D\) appearing in the statement of Lemma \ref{exlem} satisfies \(\mathbb D-\mathbb D^\top\eql\mathbb C-\mathbb C^\top\).
\end{lemma}

\begin{proof}
	Recalling \eqref{cdef2}, the covariance kernel of \(\mathbb B-\mathbb B^\top\) is given by
	\begin{align*}
	&\quad\cov((\mathbb B-\mathbb B^\top)(u,v),(\mathbb B-\mathbb B^\top)(u',v'))\\
	&=C(u\wedge u',v\wedge v')-C(u,v)C(u',v')-C(u\wedge v',v\wedge u')+C(u,v)C(v',u')\\
	&\quad-C(v\wedge u',u\wedge v')+C(v,u)C(u',v')+C(v\wedge v',u\wedge u')-C(v,u)C(v',u').
	\end{align*}
	When \(C=C^\top\), this expression simplifies to
	\begin{equation*}
	\cov((\mathbb B-\mathbb B^\top)(u,v),(\mathbb B-\mathbb B^\top)(u',v'))=2C(u\wedge u',v\wedge v')-2C(u\wedge v',v\wedge u'),
	\end{equation*}
	and further, the covariance kernel of \(\mathbb A\) given in \eqref{kerAex} simplifies to
	\begin{equation*}
	\cov(\mathbb A(u,v),\mathbb A(u',v'))=\frac{1}{2}C(u\wedge u',v\wedge v')-\frac{1}{2}C(u\wedge v',v\wedge u').
	\end{equation*}
	We therefore have \(\mathbb A\eql\frac{1}{2}(\mathbb B-\mathbb B^\top)\) when \(C=C^\top\). Observe that
	\begin{align*}
	\Phi'_C(\mathbb B-\mathbb B^\top)(u,v)&=\mathbb B(u,v)-\mathbb B(v,u)-\dot{C}_1(u,v)(\mathbb B(u,1)-\mathbb B(1,u))\\
	&\quad-\dot{C}_2(u,v)(\mathbb B(1,v)-\mathbb B(v,1))\\
	&=\mathbb C(u,v)-\mathbb C(v,u),
	\end{align*}
	where we use \(C=C^\top\) to obtain the second equality. It follows that
	\begin{equation*}
	\mathbb D=\Phi'_C\mathbb A\eql\frac{1}{2}\Phi'_C(\mathbb B-\mathbb B^\top)=\frac{1}{2}(\mathbb C-\mathbb C^\top)
	\end{equation*}
	when \(C=C^\top\). But \(\frac{1}{2}(\mathbb C-\mathbb C^\top)-\frac{1}{2}(\mathbb C-\mathbb C^\top)^\top=\mathbb C-\mathbb C^\top\), and so our claim is proved.
\end{proof}

\begin{proof}[Proof of Theorem \ref{exthm}]
	When \(W_n=R_n\) or \(W_n=T_n\) we may easily deduce that \(W_n^\tau\cwconv W\) from Lemma \ref{exlem} by applying the conditional continuous mapping theorem. When \(W_n=S_n\) a more sophisticated argument is required, but since this argument is very similar to the proof of Proposition 3 of \citet{GNQ12}, we will be terse. From Lemma \ref{exlem} above and the conditional continuous mapping theorem we have
	\begin{equation}\label{product1}
	\sqrt{n}\left(\left(\begin{array}{c}\sqrt{n}(C_n^\tau-C_n^{\tau\top})^2\\C_n^\tau\end{array}\right)-\left(\begin{array}{c}0\\\frac{1}{2}(C_n+C_n^\top)\end{array}\right)\right)\cwconv\left(\begin{array}{c}(\mathbb D-\mathbb D^\top)^2\\\mathbb D\end{array}\right)
	\end{equation}
	in the product space \(\ell^\infty([0,1]^2)\times\ell^\infty([0,1]^2)\). From the weak convergence \(\sqrt{n}(C_n-C)\rightsquigarrow\mathbb C\) and continuous mapping theorem we also have
	\begin{equation}\label{product2}
	\sqrt{n}\left(\left(\begin{array}{c}0\\\frac{1}{2}(C_n+C_n^\top)\end{array}\right)-\left(\begin{array}{c}0\\\frac{1}{2}(C+C^\top)\end{array}\right)\right)\rightsquigarrow\left(\begin{array}{c}0\\\frac{1}{2}(\mathbb C+\mathbb C^\top)\end{array}\right).
	\end{equation}
	We may use \eqref{product1} and \eqref{product2} to justify an application of the conditional delta method using the same operator as in \citet[p.\ 831]{GNQ12}, Hadamard differentiability of which is given by Lemma 4.3 of \citet{CI10}. This leads us to conclude that
	\begin{equation*}
	n\iint(C_n^\tau-C_n^{\tau\top})^2\mathrm{d}C_n^\tau\cwconv\frac{1}{2}\iint(\mathbb D-\mathbb D^\top)^2\mathrm{d}(C+C^\top);
	\end{equation*}
	that is, \(W_n^\tau\cwconv W\). The second part of Theorem \ref{exthm} follows from Lemma \ref{exlem2}.
\end{proof}

\begin{proof}[Proof of Lemma \ref{rslem}]
	Define \(\tilde{C}_n^\tau\) as in \eqref{Ftauex} but with \(\pi^0\) and \(\pi^1\) defined as in \eqref{rspi}. Observe that \(\sqrt{n}(\tilde{C}_n^\tau-\frac{1}{2}(C_n+D_n))=n^{-1/2}\sum_{i=1}^nZ_{ni}\), with the summands \(Z_{ni}\) defined as in Lemma \ref{lemcwc}. From Lemmas \ref{lemcwc} and \ref{lemeta} we have
	\begin{equation}\label{cwcrspf}
	\sqrt{n}\left(\check{C}_n^\tau-\frac{1}{2}(C_n+D_n)\right)\cwconv\mathbb A.
	\end{equation}
	From the weak convergence \(\sqrt{n}(C_n-C)\rightsquigarrow\mathbb C\), the continuous mapping theorem, and the fact that \(|D_n-C_n^\srv|\leq 4n^{-1}\) a.s.\ by Lemma \ref{gnlem}, we also have
	\begin{align}
	\sqrt{n}\left(\frac{1}{2}(C_n+D_n)-\frac{1}{2}(C+C^\srv)\right)&=\frac{1}{2}\left(\sqrt{n}(C_n-C)+\sqrt{n}(C_n-C)^\rot\right)+\frac{1}{2}\sqrt{n}(D_n-C_n^\srv)\notag\\
	&\rightsquigarrow\frac{1}{2}(\mathbb C+\mathbb C^\rot).\label{cwcrspf2}
	\end{align}
	
	We would like to use \eqref{cwcrspf} and \eqref{cwcrspf2} to apply the conditional delta method using B\"{u}cher and Volgushev's result on the Hadamard differentiability of \(\Phi\), similar to  what we did at the analogous point in the proof of Lemma \ref{exlem}. However, there is a small technical issue to resolve: the bivariate cdfs \(\check{C}_n^\tau\) and \(\frac{1}{2}(C_n+D_n)\) do not have margins grounded at zero, and so do not lie in the domain of \(\Phi\). In fact, while \(C_n(1,0)=C_n(0,1)=0\) as desired, we have \(D_n(1,0)=D_n(0,1)=n^{-1}\) a.s., and \(\check{C}_n^\tau(1,0)\) and \(\check{C}_n^\tau(0,1)\) are equal to either \(0\) or \(n^{-1}\) a.s. Letting \(\tilde{\mathbf D}\) denote the set of bivariate cdfs on \([0,1]^2\), we define a sequence of maps \(\Lambda_n:\tilde{\mathbf D}\to\mathbf D_\Phi\) by \(\Lambda_n(H)(u,v)=\mathbbm{1}(u\wedge v\geq n^{-1})H(u,v)\).	It is easy to see that the four terms
	\begin{multline*}
	\left\vert\Lambda_n(\check{C}^\tau_n)-\check{C}^\tau_n\right\vert,\quad\left\vert\Lambda_n\left(\frac{1}{2}(C_n+D_n)\right)-\frac{1}{2}(C_n+D_n)\right\vert,\\\left\vert\Phi\left(\Lambda_n(\check{C}^\tau_n)\right)-\Phi(\check{C}^\tau_n)\right\vert\quad\text{and}\quad\left\vert\Phi\left(\Lambda_n\left(\frac{1}{2}(C_n+D_n)\right)\right)-\Phi\left(\frac{1}{2}(C_n+D_n)\right)\right\vert
	\end{multline*}
	are all bounded a.s.\ by a constant multiple of \(n^{-1}\). We may therefore modify the convergences \eqref{cwcrspf} and \eqref{cwcrspf2} to obtain
	\begin{equation}\label{cwcrspf3}
	\sqrt{n}\left(\Lambda_n(\check{C}_n^\tau)-\Lambda_n\left(\frac{1}{2}(C_n+D_n)\right)\right)\cwconv\mathbb A.
	\end{equation}
	and
	\begin{equation}\label{cwcrspf4}
	\sqrt{n}\left(\Lambda_n\left(\frac{1}{2}(C_n+D_n)\right)-\frac{1}{2}(C+C^\srv)\right)\rightsquigarrow\frac{1}{2}(\mathbb C+\mathbb C^\rot).
	\end{equation}
	Theorem 2.4 of \citet{BV13} establishes that \(\Phi\) is Hadamard differentiable at \(\frac{1}{2}(C+C^\srv)\) tangentially to \(\mathbf D_0\), with derivative in direction \(h\in\mathbf D_0\) given by
	\begin{align*}
	\Phi'_{\frac{1}{2}(C+C^\srv)}h(u,v)&=h(u,v)-\frac{1}{2}\left(\dot{C}_1(u,v)+1-\dot{C}_1(1-u,1-v)\right)h(u,1)\notag\\&\quad-\frac{1}{2}\left(\dot{C}_2(u,v)+1-\dot{C}_2(1-u,1-v)\right)h(1,v).\label{ddef1}
	\end{align*}
	Note that \(\mathbb A\) concentrates on \(\mathbf D_0\), and that \(\mathbb D=\Phi'_{\frac{1}{2}(C+C^\srv)}\mathbb A\). In view of \eqref{cwcrspf3} and \eqref{cwcrspf4} we may therefore apply the conditional delta method to obtain
	\begin{equation*}
	\sqrt{n}\left(\Phi\left(\Lambda_n(\check{C}_n^\tau)\right)-\Phi\left(\Lambda_n\left(\frac{1}{2}(C_n+D_n)\right)\right)\right)\cwconv\mathbb D.
	\end{equation*}
	Consequently,
	\begin{equation*}
	\sqrt{n}\left(\Phi(\check{C}_n^\tau)-\Phi\left(\frac{1}{2}(C_n+D_n)\right)\right)\cwconv\mathbb D.
	\end{equation*}	
	As in the proof of Lemma \ref{exlem}, \(\Phi(\check{C}_n^\tau)\) differs from \(C_n^\tau\) and \(\Phi(\frac{1}{2}(C_n+D_n))\) differs from \(\frac{1}{2}(C_n+D_n)\) by no more than \(2n^{-1}\) a.s. Thus we have \eqref{rscwc2} as claimed.
\end{proof}

\begin{lemma}\label{rslem2}
	Let \(\mathbb D\) be the random element appearing in the statement of Lemma \ref{rslem}. If \(C=C^\srv\) then there exist centered Gaussian random elements \(\mathbb E\) and \(\mathbb F\) of \(\ell^\infty([0,1]^2)\) that are equal in law and satisfy
	\begin{align*}
	\mathbb C(u,v)-\mathbb C(1-u,1-v)&=\mathbb E(u,v)-\dot{C}_1(u,v)\mathbb E(u,1)-\dot{C}_2(u,v)\mathbb E(1,v),\\\mathbb D(u,v)-\mathbb D(1-u,1-v)&=\mathbb F(u,v)-\dot{C}_1(u,v)\mathbb F(u,1)-\dot{C}_2(u,v)\mathbb F(1,v).
	\end{align*}
	Consequently, \(\mathbb D-\mathbb D^\rot\eql\mathbb C-\mathbb C^\rot\). The covariance kernel of \(\mathbb E\) and \(\mathbb F\) is
	\begin{equation*}
	\cov(\mathbb E(u,v),\mathbb E(u',v'))=2C(u\wedge u',v\wedge v')-2\pr(1-u'<U\leq u,1-v'<V\leq v).
	\end{equation*}
\end{lemma}

\begin{proof}	
	Let \(\Psi:\ell^\infty([0,1]^2)\to\ell^\infty([0,1]^2)\) be the map
	\begin{equation*}
	\Psi(\theta)(u,v)=\theta(u,v)-\theta(1-u,1-v)+\theta(1-u,1)+\theta(1,1-v),
	\end{equation*}
	and define \(\mathbb E=\Psi(\mathbb B)\) and \(\mathbb F=\Psi(\mathbb A)\), where \(\mathbb A\) is as defined in the statement of Lemma \ref{rslem}. The operator \(\Psi\) is continuous and linear, and the random elements \(\mathbb A\) and \(\mathbb B\) have continuous sample paths and therefore take values in a separable subset of \(\ell^\infty([0,1]^2)\). It follows from Proposition 3.7.2 of \citet{B98} that \(\mathbb E\) and \(\mathbb F\) are centered Gaussian random elements of \(\ell^\infty([0,1]^2)\). When \(C=C^\srv\) we have \(\dot{C}_j=1-\dot{C}_j^\rot\) for \(j=1,2\). Therefore, from the definition of \(\mathbb C\) given in \eqref{cdef1},
	\begin{align*}
	\mathbb C(u,v)-\mathbb C(1-u,1-v)&=\mathbb B(u,v)-\mathbb B(1-u,1-v)+\mathbb B(1-u,1)+\mathbb B(1,1-v)\\&\quad-\dot{C}_1(u,v)(\mathbb B(u,1)+\mathbb B(1-u,1))\\&\quad-\dot{C}_2(u,v)(\mathbb B(1,v)+\mathbb B(1,1-v))\\
	&=\mathbb E(u,v)-\dot{C}_1(u,v)\mathbb E(u,1)-\dot{C}_2(u,v)\mathbb E(1,v).
	\end{align*}
	Similarly, from the definition of \(\mathbb D\) given in the statement of Lemma \ref{rslem},
	\begin{equation*}
	\mathbb D(u,v)-\mathbb D(1-u,1-v)=\mathbb F(u,v)-\dot{C}_1(u,v)\mathbb F(u,1)-\dot{C}_2(u,v)\mathbb F(1,v).
	\end{equation*}
	
	It remains only to show that \(\mathbb E\) and \(\mathbb F\) have the claimed covariance kernel. First we examine the covariance kernel of \(\mathbb E\). Let \(U_i=F(X_i)\) and \(V_i=G(Y_i)\). Since \(\mathbb B(u,v)\) is the weak limit of \(n^{-1/2}\sum_{i=1}^n(\mathbbm{1}(U_i\leq u,V_i\leq v)-C(u,v))\), the continuous mapping theorem implies that \(\mathbb E(u,v)\) is the weak limit of \(n^{-1/2}\sum_{i=1}^n\Psi(\mathbbm{1}(U_i\leq u,V_i\leq v)-C(u,v))\). The terms in this latter summation are centered and iid so, using the fact that \(\Psi(C)=1\) when \(C=C^\srv\), we find that
	\begin{equation}\label{ecov}
	\cov\left(\mathbb E(u,v),\mathbb E(u',v')\right)=\ex\bigl((\Psi(\mathbbm{1}(U_i\leq u,V_i\leq v))-1)(\Psi(\mathbbm{1}(U_i\leq u',V_i\leq v'))-1)\bigr).
	\end{equation}
	Observe that
	\begin{align*}
	\mathbbm{1}(U_i\leq1-u,V_i\leq1-v)&=\mathbbm{1}(U_i\leq1-u)+\mathbbm{1}(V_i\leq1-v)-1\notag\\&\quad+\mathbbm{1}(1-U_i<u,1-V_i<v).
	\end{align*}
	It follows that
	\begin{align}\label{neattrick2}
	\Psi(\mathbbm{1}(U_i\leq u,V_i\leq v))-1&=\mathbbm{1}(U_i\leq u,V_i\leq v)-\mathbbm{1}(U_i\leq1-u,V_i\leq1-v)\notag\\
	&\quad+\mathbbm{1}(U_i\leq1-u)+\mathbbm{1}(V_i\leq1-v)-1\notag\\
	&=\mathbbm{1}(U_i\leq u,V_i\leq v)-\mathbbm{1}(1-U_i<u,1-V_i<v).
	\end{align}
	It now follows from \eqref{ecov} that
	\begin{align}
	\cov\left(\mathbb E(u,v),\mathbb E(u',v')\right)&=\ex\big{(}(\mathbbm{1}(U_i\leq u,V_i\leq v)-\mathbbm{1}(1-U_i<u,1-V_i<v))\notag\\
	&\quad\times(\mathbbm{1}(U_i\leq u',V_i\leq v')-\mathbbm{1}(1-U_i<u',1-V_i<v'))\big{)}\notag\\
	&=\pr\left(U_i\leq u\wedge u',V_i\leq v\wedge v'\right)\notag\\&\quad+\pr\left(1-U_i<u\wedge u',1-V_i<v\wedge v'\right)\notag\\
	&\quad-\pr\left(1-u'<U_i\leq u,1-v'<V_i\leq v\right)\notag\\
	&\quad-\pr\left(1-u<U_i\leq u',1-v<V_i\leq v'\right)\notag\\
	&=2C\left(u\wedge u',v\wedge v'\right)-2\pr\left(1-u'<U_i\leq u,1-v'<V_i\leq v\right),\label{ecov2}
	\end{align}
	where to obtain the last equality we use the fact that \((U_i,V_i)\) and \((1-U_i,1-V_i)\) both have continuous joint cdf \(C\) when \(C=C^\srv\).
	
	Next we examine the covariance kernel of \(\mathbb F\). By arguing as we did at the beginning of the proof of Lemma \ref{rslem}, we deduce from Lemmas \ref{lemcwc} and \ref{gnlem} that
	\begin{equation*}
	\sqrt{n}\left(\tilde{C}_n^\tau-\frac{1}{2}(C_n+C_n^\srv)\right)\cwconv\mathbb A.
	\end{equation*}
	Using the conditional continuous mapping theorem we then obtain
	\begin{equation*}
	\sqrt{n}\left(\Psi(\tilde{C}_n^\tau)-\frac{1}{2}\Psi(C_n+C_n^\srv)\right)\cwconv\mathbb F.
	\end{equation*}
	It is simple to verify that	\(\Psi(C_n+C_n^\srv)(u,v)=u+v+C_n(1-u,1)+C_n(1,1-v)\). Since \(C_n(1-u,1)\) and \(C_n(1,1-v)\) differ by no more than \(n^{-1}\) from \(1-u\) and \(1-v\) respectively a.s., it follows that \(\Psi(C_n+C_n^\srv)\) differs from \(2\) by no more than \(2n^{-1}\) a.s. Thus
	\begin{equation}\label{fconv}
	\sqrt{n}(\Psi(\tilde{C}_n^\tau)-1)\cwconv\mathbb F.
	\end{equation}
	By applying \eqref{neattrick2} with \(\pi^{\tau_i}(U_{ni},V_{ni})\) in place of \((U_i,V_i)\), we find that
	\begin{equation}\label{fconv2}
	\sqrt{n}(\Psi(\tilde{C}_n^\tau)-1)=\frac{1}{\sqrt{n}}\sum_{i=1}^n\bigl(\Psi(\mathbbm{1}(\pi^{\tau_i}(U_{ni},V_{ni})\leq(u,v))-1\bigr)=\frac{1}{\sqrt{n}}\sum_{i=1}^nW_{ni}(u,v),
	\end{equation}
	where we define
	\begin{equation*}
	W_{ni}(u,v)=\mathbbm{1}(\pi^{\tau_i}(U_{ni},V_{ni})\leq(u,v))-\mathbbm{1}(\pi^{\tau_i}(U_{ni},V_{ni})>(1-u,1-v)).
	\end{equation*}
	The summands \(Z_{ni}\) appearing in the statement of Lemma \ref{lemcwc} satisfy
	\begin{equation*}
	2Z_{ni}(u,v)=\mathbbm{1}(\pi^{\tau_i}(U_{ni},V_{ni})\leq(u,v))-\mathbbm{1}(\pi^{\tau_i}(U_{ni},V_{ni})\geq(1-u,1-v)).
	\end{equation*}
	Observe that \(\vert W_{ni}(u,v)-2Z_{ni}(u,v)\vert\) is bounded by \(4\) a.s.\ after summing over \(i=1,\ldots,n\). This follows from the fact that the first component of \(\pi^{\tau_i}(U_{ni},V_{ni})\) is equal to \(1-u\) for at most two values of \(i\) a.s., and the second component of \(\pi^{\tau_i}(U_{ni},V_{ni})\) is equal to \(1-v\) for at most two values of \(i\) a.s. Consequently, from \eqref{fconv} and \eqref{fconv2}, we have
	\begin{equation*}
	\frac{1}{\sqrt{n}}\sum_{i=1}^nZ_{ni}\cwconv\frac{1}{2}\mathbb F.
	\end{equation*}
	It now follows from Lemma \ref{lemcwc} that \(\mathbb F\eql2\mathbb A\). When \(C=C^\srv\) the covariance kernel of \(\mathbb A\) given in \eqref{kerArs} simplifies to
	\begin{equation*}
	\cov(\mathbb A(u,v),\mathbb A(u',v'))=\frac{1}{2}C(u\wedge u',v\wedge v')-\frac{1}{2}\pr(1-u'<U_i\leq u,1-v'<V_i\leq v).
	\end{equation*}
	The covariance kernel of \(\mathbb F\) is \(4\) times the covariance kernel of \(\mathbb A\), and thus equal to the covariance kernel of \(\mathbb E\) given in \eqref{ecov2}. Thus \(\mathbb E\eql\mathbb F\).
\end{proof}

\begin{proof}[Proof of Theorem \ref{rsthm}]
	Lemma \ref{gnlem} implies that
	\begin{equation*}
	\left\vert C_n(u,v)+D_n(u,v)-C_n(1-u,1-v)-D_n(1-u,1-v)-2u-2v+2\right\vert\leq 8n^{-1}
	\end{equation*}
	a.s., and similarly \(|D_n^\tau-C_n^{\tau\srv}|\leq4n^{-1}\) a.s. It follows easily that
	\begin{equation}\label{shortpf}
	\left\vert\left(C_n^\tau-\frac{1}{2}\left(C_n+D_n\right)\right)-\left(C_n^\tau-\frac{1}{2}\left(C_n+D_n\right)\right)^\rot-\left(C_n^\tau-D_n^\tau\right)\right\vert\leq8n^{-1}.
	\end{equation}
	From Lemma \ref{rslem} and the conditional continuous mapping theorem we thus have \(\sqrt{n}(C_n^\tau-D_n^\tau)\cwconv\mathbb D-\mathbb D^\rot\). When \(W_n=R_n'\) or \(W_n=T_n'\), another application of the conditional continuous mapping theorem yields \(W_n\cwconv W\). When \(W_n=S_n'\), we instead use Lemma \ref{rslem}, inequality \eqref{shortpf} and the conditional continuous mapping theorem to write
	\begin{equation}\label{product1rs}
	\sqrt{n}\left(\left(\begin{array}{c}\sqrt{n}(C_n^\tau-D_n^\tau)^2\\C_n^\tau\end{array}\right)-\left(\begin{array}{c}0\\\frac{1}{2}(C_n+D_n)\end{array}\right)\right)\cwconv\left(\begin{array}{c}(\mathbb D-\mathbb D^\rot)^2\\\mathbb D\end{array}\right),
	\end{equation}
	and use \eqref{cwcrspf2}, shown in the proof of Lemma \ref{rslem}, to write
	\begin{equation}\label{product2rs}
	\sqrt{n}\left(\left(\begin{array}{c}0\\\frac{1}{2}(C_n+D_n)\end{array}\right)-\left(\begin{array}{c}0\\\frac{1}{2}(C+C^\srv)\end{array}\right)\right)\rightsquigarrow\left(\begin{array}{c}0\\\frac{1}{2}(\mathbb C+\mathbb C^\rot)\end{array}\right).
	\end{equation}
	Using \eqref{product1rs} and \eqref{product2rs}, we may apply the conditional delta method in the same way as in the first paragraph of the proof of Theorem \ref{exthm} to deduce that \(W_n\cwconv W\). The second part of Theorem \ref{rsthm} then follows from Lemma \ref{rslem2}.
\end{proof}

\end{document}